\documentclass{iopjournal}

\begin{document}
\articletype{Paper}

\title{Security bounds for unidimensional discrete-modulated CV-QKD: a Gaussian extremality approach}

\author{John A. Mora Rodríguez$^{1,2,*}$, Maron F. Anka$^1$, Leonardo J. Pereira$^{3,4}$, Micael A. Dias$^{5}$ and Alexandre B. Tacla$^{1}$}

\affil{$^1$QuIIN - Quantum Industrial Innovation, EMBRAPII CIMATEC Competence Center in Quantum Technologies, SENAI CIMATEC, Av. Orlando Gomes 1845, Salvador, BA, Brazil, CEP 41650-010.}

\affil{$^2$Instituto de Matemática, Estatística e Computação Científica, Universidade Estadual de Campinas, CEP 13083-859, Campinas, Brazil.}

\affil{$^3$Instituto de Física, Universidade Federal do Rio de Janeiro, Avenida Athos da Silveira Ramos 149, Rio de Janeiro, RJ, C. P. 68.528, CEP 21945-970, Brazil}

\affil{$^4$Instituto de Matemática, Universidade Federal do Rio de Janeiro, Avenida Athos da Silveira Ramos 149, Rio de Janeiro, RJ, C. P. 68.530, CEP 21945-970, Brazil}

\affil{$^5$Department of Electrical and Photonics Engineering, Technical University of Denmark, 2800  Lyngby, Denmark.}

\affil{$^*$Author to whom any correspondence should be addressed.}

\email{john.rodriguez@fieb.org.br}


\begin{abstract}
\justifying Unidimensional (1D) Gaussian-modulated continuous-variable quantum key distribution protocols have been proposed as a way to simplify implementation and reduce costs through single-quadrature modulation, requiring only one modulator while maintaining compatibility with standard optical infrastructure. Here, we determine security bounds for a 1D discrete-modulated protocol under the Gaussian extremality assumption by extending the method of Ghorai \textit{et al.} [Phys. Rev. X 9, 021059 (2019)]. We establish the appropriate symmetry arguments to extend the method to the 1D discrete-modulated case, define the physicality zone in which the protocol is allowed to operate, and derive a security bound against collective attacks in the asymptotic regime via semidefinite programming. Our analysis for uniformly distributed coherent states reveals a limitation: the Gaussian extremality assumption systematically overestimates Eve's information with increasing constellation size, yielding bounds so conservative that secure key extraction becomes increasingly limited for larger constellations, even under ideal conditions. This overestimation worsens with excess noise and restricts viable modulation amplitudes to near-vacuum values. Unlike two-dimensional (2D) protocols, where Gaussian extremality improves with constellation size, 1D protocols lack the growing phase-space isotropy required for the approximation to remain tight as the constellation grows. Our results expose these limitations and highlight the necessity of alternative methods or optimized non-uniform constellation designs for this class of protocols.
\end{abstract}


\section{Introduction}
\label{introduction}

Quantum key distribution (QKD) is one of the most mature technological applications of quantum information science, enabling information-theoretic secure key distribution between two trusted parties \cite{pirandola2020advances,zhang2024continuous,usenko2026continuous,anka2026introductory,kubala2026advanced}. In particular, continuous-variable (CV) QKD stands out as a promising candidate for real-world implementation due to its compatibility with current optical telecommunications infrastructure \cite{zhang2024continuous}. This off-the-shelf technology enables the use of standard apparatus, such as homodyne detectors, thereby reducing the implementation cost of CV-QKD systems compared to their discrete-variable QKD counterparts, which rely on (more expensive) single-photon detectors \cite{pirandola2020advances}. 

There are basically two families of CV-QKD protocols: Gaussian and discrete modulation formulations. Even though Gaussian-modulated (GM) protocols are the most advanced regarding security analyses \cite{diamanti2015distributing}, their experimental realization faces challenges in both hardware and classical post-processing. The near-continuous modulation requires high-precision modulators and makes information reconciliation computationally demanding due to the very large alphabet size. The development of discrete-modulated (DM) CV-QKD protocols addresses these issues by employing finite constellations of Gaussian states with the goal of approximating Gaussian modulation. This approach significantly simplifies classical post-processing, while maintaining compatibility with practical experimental implementations \cite{leverrier2009unconditional,becir2012continuous,bradler2018security}. The flexibility in constellation design, including phase-shift keying (PSK), amplitude and phase-shift keying (APSK), and quadrature amplitude modulation (QAM), allows for optimization of secure key generation efficiency under realistic channel conditions while maintaining compatibility with coherent optical communication standards \cite{almeida2021secret,almeida2023reconciliation,diamanti2024,pan2022experimental, liao2025continuous}.

However, proving the security of DM protocols presents greater complexity than in the Gaussian-modulated case. In GM protocols, the optimal eavesdropping attack is known to be Gaussian \cite{Garcia2006,Navascues2006}, which significantly simplifies the security analysis. For DM protocols, the optimal attack remains unknown, requiring more sophisticated approaches that do not rely on specific channel models---a departure from earlier methods that assumed simplified scenarios such as pure-loss \cite{Heid2006,Sych2010,Kanitschar2022} or linear channels \cite{Leverrier2011, papanastasiou2018}.

Early security analyses for DM protocols were developed considering only simplified constellations with few states \cite{bradler2018security,Leverrier2011,Zhao2009}. The first attempt to establish a treatment for general constellations was developed in 2019 by Ghorai \textit{et al.} for QPSK modulation \cite{Ghorai2019}, which was later generalized by Denys \textit{et al.} in Ref.\cite{Denys2021} to arbitrary constellation shapes and sizes under the Gaussian extremality assumption \cite{wolf2006}. In that same year, Lin, Upadhyaya, and Lütkenhaus developed an alternative method, valid for arbitrary constellations, that avoids the Gaussian extremality argument and achieves tighter bounds closer to the GM case \cite{lin2019asymptotic}. However, this method comes with a significantly higher computational cost, especially for constellations beyond QPSK \cite{usenko2025continuous}.

In this work, we extend the method of Ghorai \textit{et al.} \cite{Ghorai2019} to unidimensional (1D) discrete modulation of coherent states, relying on the Gaussian extremality assumption. While this approach generally overestimates the eavesdropper's capabilities and may thus yield looser security bounds than the numerical method of Lin \textit{et al.} \cite{lin2019asymptotic}, its computational tractability makes it a valuable framework for rapidly establishing security bounds across different constellation designs and identifying configurations that lead to simpler and more practical implementations. In particular, 1D modulation of coherent states offers the experimental advantage of requiring only a single modulator, reducing both system complexity and cost compared to conventional two-dimensional (2D) modulation schemes \cite{Usenko2015}. We note that during the completion of this article, Wu \textit{et al.} \cite{wu2026discretely} independently published a security analysis of 1D DM coherent-state protocols using the method of Lin \textit{et al.}, achieving tighter bounds in both the asymptotic and finite-size regimes while confirming the practical relevance of this class of protocols. However, this method is considerably more computationally expensive \cite{usenko2025continuous} than using the Gaussian extremality theorem, which motivates further investigation of more efficient approaches. Our complementary approach thus provides insights into the limitations of the Gaussian extremality assumption in this setting.

The first 1D GM CV-QKD protocol was introduced by Usenko and Grosshans in Ref.\cite{Usenko2015}, in which only the $\hat{q}$ quadrature is modulated, unlike conventional schemes that modulate both quadratures symmetrically. The authors established the security of the protocol against collective attacks for general phase-sensitive Gaussian channels, motivated by key practical advantages: simplified hardware (one modulator instead of two), no need for high-extinction-ratio devices, and reduced implementation costs. The first attempt to extend this protocol to discrete modulation was presented in Ref.\cite{zhao2020unidimensional}. However, there is a possible inadequacy in the formalism used in their extension to the discrete-modulated protocol, which is evidenced by the fact that their result achieved a higher secret key rate than the 1D GM CV-QKD protocol. In this work, we provide the correct mathematical framework for 1D DM coherent-state protocols by establishing the appropriate symmetry arguments for states symmetrically and uniformly distributed along the real line in phase space. We establish a security bound for the protocol against collective attacks via semidefinite programming (SDP), under the Gaussian extremality assumption in the asymptotic regime, and we explicitly characterize the physicality zone in which the protocol operates. Crucially, we demonstrate that the Gaussian extremality assumption — while computationally tractable — systematically overestimates Eve's information in this setting, revealing limitations of this widely used approach for such 1D DM protocols.

This paper is organized as follows. In Sec.\ref{1dgauss}, we briefly review the original unidimensional Gaussian-modulated CV-QKD protocol and introduce the elements required to extend it to the discrete modulation setting. In Sec.\ref{dm1d}, we formulate the discrete-modulated protocol, exploiting symmetry properties and invoking Gaussian extremality arguments to characterize its security. We also identify the physicality region that ensures the protocol remains experimentally implementable. In Sec.\ref{lowerboundsdp}, we present the SDP formulation used to derive security bounds, and in Sec.\ref{numerical}, we discuss the corresponding numerical results. Finally, our concluding remarks are given in Sec.\ref{conclusion}.

\section{1D Gaussian-modulated protocol}
\label{1dgauss}

 Unidimensional CV-QKD protocols align with the broader goal of simplifying CV-QKD implementations to facilitate real-world use. The 1D GM protocol description in Ref.\cite{Usenko2015} is as follows: Alice generates coherent states by modulating exclusively one quadrature ($\hat{q}$, corresponding to the amplitude quadrature), using a Gaussian random variable with variance $V_{\text{Mod}}$. These states are typically transmitted through a phase-sensitive channel. Bob performs homodyne detection, predominantly measuring the $\hat{q}$ quadrature and sporadically switching to the $\hat{p}$ quadrature in order to obtain statistics about the channel parameters. After an adequate number of runs, Alice and Bob evaluate the security and derive a secret key from the $\hat{q}$ quadrature data through a reverse reconciliation procedure.

The security proof of this protocol is based on the Gaussian extremality theorem \cite{wolf2006} and includes an additional physicality verification step whenever Eve's interference on the unmodulated quadrature cannot be determined. In this regard, the entanglement-based (EB) protocol equivalence is indispensable. In the equivalent EB version of the protocol, it is assumed that Alice prepares a two-mode squeezed vacuum state with variance $V=\sqrt{V_{\text{Mod}}+1}$ and covariance matrix given by
\begin{equation}
    \gamma = \begin{pmatrix}
V & 0 & \sqrt{V^2-1} & 0 \\
0 & V & 0 & -\sqrt{V^2-1} \\
\sqrt{V^2-1} & 0 & V & 0 \\
0 & -\sqrt{V^2-1} & 0 & V
\end{pmatrix}.
\end{equation} 

Alice then measures one of the modes of this state and sends the other mode to Bob. In the case of 2D Gaussian modulation of coherent states, this process is effectively implemented by a heterodyne measurement. However, for the 1D GM protocol, Alice performs a homodyne measurement, which causes the state sent to Bob to collapse into a squeezed state with covariance matrix $\text{diag}(1/V,V)$. To preserve the equivalence with the prepare-and-measure (P\&M) protocol, the state sent to Bob must be a coherent state, which implies that it is necessary to apply squeezing to Bob's mode with the squeezing parameter $-\log(\sqrt{V})$. In this sense, the covariance matrix of the purification in the EB scheme for unidimensional Gaussian modulation is given by
\begin{equation}
    \gamma_{1D} = \begin{pmatrix}
V & 0 & \sqrt{V(V^2-1)} & 0 \\
0 & V & 0 & -\sqrt{\frac{V^2-1}{V}} \\
\sqrt{V(V^2-1)} & 0 & V^2 & 0 \\
0 & -\sqrt{\frac{V^2-1}{V}} & 0 & 1
\end{pmatrix}.
\end{equation} 

After Bob's system is sent through a phase-sensitive Gaussian channel, the state shared by Alice and Bob has an associated covariance matrix given by
    \begin{align}\nonumber
        \gamma_{AB} = 
        \begin{pmatrix}
\sqrt{1 + V_M} & 0 & \sqrt{T_q V_M} (1 + V_M)^{\frac{1}{4}} & 0 \\
0 & \sqrt{1 + V_M} & 0 & C_p \\
\sqrt{T_q V_M} (1 + V_M)^{\frac{1}{4}} & 0 & 1 + T_q(V_M + \xi_q) & 0 \\
0 & C_p & 0 & V_p 
\end{pmatrix},
    \end{align}
where $T_q$ and $\xi_q$ are, respectively, the channel transmittance and excess noise in the $\hat{q}$ quadrature. The parameter $V_p$ is the output variance of Bob's mode in the $\hat{p}$ quadrature. The parameter $C_p$ is the correlation between Alice's and Bob's modes in the $\hat{p}$ quadrature, which is unknown since Alice does not modulate this quadrature.

The key rate is calculated for the Gaussian state $\rho_{AB}$ with covariance matrix $\gamma_{AB}$ through the computation of its symplectic eigenvalues and the symplectic eigenvalue of the covariance matrix conditioned on Bob's homodyne measurement \cite{anka2026introductory,Laudenbach2018}. The parameter $C_p$ is taken as the parameter within the physicality region determined by the Heisenberg uncertainty principle \cite{Serafini_2005} that minimizes the key rate. 

\subsection{Extending to Discrete Modulation}\label{sec: extending}

When extending the 1D protocol to discrete modulation, it is necessary to establish a series of key points in the security analysis. The most essential is that Gaussian attacks are not provably optimal in the case of discrete modulation; therefore, a security proof should not depend on a specific channel model. Additionally, if we are going to use the Gaussian extremality theorem as a tool in the security proof, the structure of the covariance matrix of the state shared by Alice and Bob\footnote{The shared state between Alice and Bob $\rho_{AB}$ in the EB scheme is only a mathematical object; Alice and Bob do not have access to it. Additionally, this state depends on the quantum channel. Since the real channel is unknown, the state is also unknown. The objective in the security proof is to optimize the information accessible to Eve over all states $\rho_{AB}$ compatible with the measurement outcome statistics.} in the EB scheme must be modified through an appropriate symmetrization process \cite{leverrier:tel-00451021} in order to simplify the analysis without significantly impacting the secret key rate (SKR). Note that in the case of 2D modulations, the usual symmetrization is performed over the group of Gaussian unitary operators corresponding to real symplectic orthogonal transformations in phase space \cite{leverrier:tel-00451021}. This ensures that the covariance matrix of the state\footnote{In the collective attack scenario where the global state of the protocol is given by $\rho_{AB}^{\otimes n}$.} $\rho_{AB}$ has the following structure
\begin{equation}\label{CM-SYM}
    \gamma_{AB}^{sym}=\begin{pmatrix}
        V\mathds{1} & Z\sigma_Z\\
        Z\sigma_Z & W\mathds{1}
    \end{pmatrix},
\end{equation}
where
\begin{align*}
V &:= \frac{1}{2} (\langle \hat{q}_A^2 \rangle_{\rho_{AB}} + \langle \hat{p}_A^2 \rangle_{\rho_{AB}}), \\
W &:= \frac{1}{2} (\langle \hat{q}_B^2 \rangle_{\rho_{AB}} + \langle \hat{p}_B^2 \rangle_{\rho_{AB}}), \\
Z &:= \frac{1}{4} \left( \langle \{\hat{q}_A, \hat{q}_B\} \rangle_{\rho_{AB}} - \langle \{\hat{p}_A, \hat{p}_B\} \rangle_{\rho_{AB}}\right).
\end{align*}

Applying this symmetrization process to the state $\rho_{AB}$ makes it isotropic in phase space, enforcing symmetries that are characteristic of a Gaussian state\footnote{This does not imply that the resulting state is Gaussian, only that its statistics become similar to those produced by a Gaussian state \cite{Leverrier2012}.}. An important observation is that this symmetrization process can be equivalently implemented through an orthogonal transformation applied directly to Alice and Bob's classical data, a considerably simpler procedure from a practical standpoint. However, the equivalence between these two approaches depends crucially on one detail: Alice and Bob cannot retain information about which specific transformation was applied. In fact, as observed by Leverrier \cite{leverrier:tel-00451021}, it is not necessary for Alice and Bob to physically execute these transformations. The essential requirement is that they must ensure that no information about the original structure or ordering of the data is used in the subsequent steps of the protocol. When this condition is satisfied, the protocol becomes indistinguishable from a scenario where random transformations would have been applied and subsequently forgotten.

The state shared by Alice and Bob in the EB picture has the form $\rho=\left( \mathds{1}_A \otimes \mathcal{E}_{A^{'} \to B} \right)\left( \, |\Phi\rangle\langle \Phi|_{A A^{'}} \, \right)$, where $\ket{\Phi}_{AA'}=\sum_k\sqrt{p_k}\ket{\psi_k}\ket{\alpha_k}$ is a purification of the average state of the constellation $\tau=\sum_kp_k\ket{\alpha_k}\bra{\alpha_k}$. Here, $\{\ket{\psi_k}\}_k$ is an orthonormal basis, $\ket{\alpha_k}$ is a coherent state, and $\mathcal{E}_{A'\rightarrow B}$ is a completely positive trace-preserving map that models the quantum channel. Note that the measurement processes are not modified by the symmetrized scheme, since Alice performs the projective measurements $\{\ket{\psi_k}\bra{\psi_k}\}_k$ on her half of the entangled state $\ket{\Phi}_{AA'}$, which is equivalent to preparing and sending the corresponding coherent state to Bob in the P\&M scheme, without requiring any squeezing transformation on Bob's mode.

Additionally, we observe that the symmetrization procedure leading to the covariance matrix structure in Eq.~\eqref{CM-SYM} does not apply to the 1D scenario. First, this symmetrization mixes information from the $\hat{q}$ and $\hat{p}$ quadratures through rotations in phase space. However, in the 1D protocol, only the $\hat q $
quadrature is modulated and measured by Bob, breaking the symmetry between quadratures. This contradicts the requirement that no information about the original structure or ordering of the data is used. Furthermore, in the covariance matrix structure of Eq.~\eqref{CM-SYM}, the parameter $Z$ (which must be optimized to apply the Gaussian extremality theorem \cite{Ghorai2019,Denys2021}) mixes correlations between both quadratures of Alice's and Bob's modes. This complicates the security analysis because the correlation in the unmodulated quadrature $\langle \hat{p}_A\hat{p}_B\rangle$ is unknown and must be bounded using the physicality constraint of the covariance matrix, analogously to the GM case~\cite{Usenko2015}.

\section{1D discrete-modulated protocol}
\label{dm1d}

In this section, we provide a detailed description of the 1D DM CV-QKD protocol and its security analysis under the Gaussian extremality assumption, employing an appropriate symmetrization procedure. In addition, we show the allowed physicality region in which the protocol can be realized.
\subsection{Protocol description}

We now present the description of the protocol with 1D discrete modulation:

\begin{enumerate}
    \item \textbf{State preparation}: Alice prepares one of \(2N\) coherent states,
\[
\{|\alpha_k\rangle, |-\alpha_k\rangle\}_{k=0}^{N-1},
\]
where each state \(|\pm\alpha_k\rangle\) is selected with probability \(p_k/2\). The amplitudes are defined such that \(\alpha_0 \in \mathbb{R}\) and \(\alpha_k = r_k \alpha_0\) for \(k = 1, 2, \ldots, N-1\) and $r_k\in \mathbb{R}^+$ (see Fig. \ref{fig:modulation}). Alice then transmits the prepared quantum state to Bob through the quantum channel.
\item \textbf{Measurement}: Upon receiving the quantum state, Bob performs a homodyne measurement on the \(\hat{q}\) quadrature.

During the protocol rounds, Bob predominantly measures the modulated \(\hat{q}\) quadrature. However, it is necessary to occasionally collect statistics about the channel properties in the unmodulated quadrature, \(\hat{p}\). Therefore, Bob must sporadically measure the \(\hat{p}\) quadrature \cite{Usenko2015}. Nonetheless, we consider the asymptotic regime, in which the occasional sampling of the $\hat{p}$ quadrature for parameter estimation has a negligible impact on the key rate \cite{Lo2005}.
\item \textbf{Parameter estimation}: After a sufficiently large number of rounds, Alice and Bob randomly select a subset of their data for channel parameter estimation. This subset is publicly announced over the classical authenticated channel and subsequently discarded from the key generation process. From this selected subset of data, Alice and Bob determine the relevant quantum channel parameters by computing the expected values of the variances of both quadratures and the covariance in the modulated quadrature, conditioned on each of the states transmitted by Alice.
These parameters are then used to constrain the set of physically admissible bipartite states $\rho_{AB}$, corresponding to the state shared by Alice and Bob under the assumption that the channel reflects Eve's optimal collective attack. The secret key rate is subsequently evaluated through the optimization problem in Eq. \eqref{sdp}. If the resulting rate vanishes, the protocol is aborted; otherwise, Alice and Bob proceed to the classical post-processing steps.
\item \textbf{Reverse reconciliation}: Alice and Bob perform reverse reconciliation \cite{Grosshans2003} on their correlated \(\hat{q}\) quadrature measurements to extract a secret key. Bob transmits additional classical information to Alice through the authenticated classical channel, enabling her to more accurately reconstruct Bob's string.\footnote{Bob's string is obtained through a discretization of his measurement outcomes. (e.g., \cite{Ghorai2019,lin2019asymptotic})} This reconstruction is performed by applying an error-correction protocol \cite{usenko2026continuous,luo2024overview}.
\item \textbf{Privacy amplification}: Alice and Bob apply a privacy-amplification protocol to generate a secret key \cite{usenko2026continuous,luo2024overview}.
\end{enumerate}
\subsection{Gaussian extremality and symmetry assumptions}
\label{assumptions}

The secure key rate is given by the Devetak-Winter bound \cite{DevetakWinter2005} 
\begin{equation}\label{key rate}
    K \geq \beta I(X;Y) - \sup \chi(Y;E),
\end{equation}
where \(I(X;Y)\) is the mutual information between Alice’s and Bob’s real-valued classical variables, \(X\) and \(Y\), respectively, \( \beta \in [0,1] \) is a parameter that quantifies the efficiency of the information reconciliation process and characterizes the performance of the error correction \cite{Yang2023}, and \(\chi(Y;E)\) is the Holevo information between Bob's variable \(Y\) and Eve’s system. The supremum is taken over all possible channels that are consistent with the statistics observed by Alice and Bob in the parameter estimation stage. 

Our approach is based on the method originally proposed by Ghorai \textit{et al.} \cite{Ghorai2019}, here adapted to the 1D protocol. This method allows us to establish, through semidefinite programming, an upper bound on the information that a malicious party can extract, consistent with the observed statistics of Alice’s and Bob’s data. By characterizing the physicality region of the problem and imposing the appropriate constraints within the SDP framework, we compute the SKR. We establish an appropriate symmetrization for the 1D protocol by considering only symmetries compatible with practical implementation. Specifically, we employ a reflection with respect to the axis corresponding to the $\hat{p}$ quadrature in phase space---equivalent to a phase conjugation or time-reversal operation. This transformation can be applied to Alice's and Bob's classical data even though it is not a symplectic operation and therefore has no direct physical realization in phase space. Nevertheless, the implementation of this operation can be simulated if Alice and Bob agree to perform the coordinate transformation $\hat{q}\mapsto -\hat{q}$ and $\hat{p}\mapsto \hat{p}$ (see Ref.\cite[p.~133]{leverrier:tel-00451021}).

Any other orthogonal operation would require information about the unmodulated quadrature, which is not accessible in this setting. As a result, we use a discrete group composed of the identity and the reflection operation in the symmetrization procedure \cite{Leverrier2012}. Consequently, the covariance matrix associated with the state shared by Alice and Bob in the EB version of the protocol \cite{Ghorai2019,Denys2021},  
\begin{center}
    \includegraphics[width=1.0\textwidth]{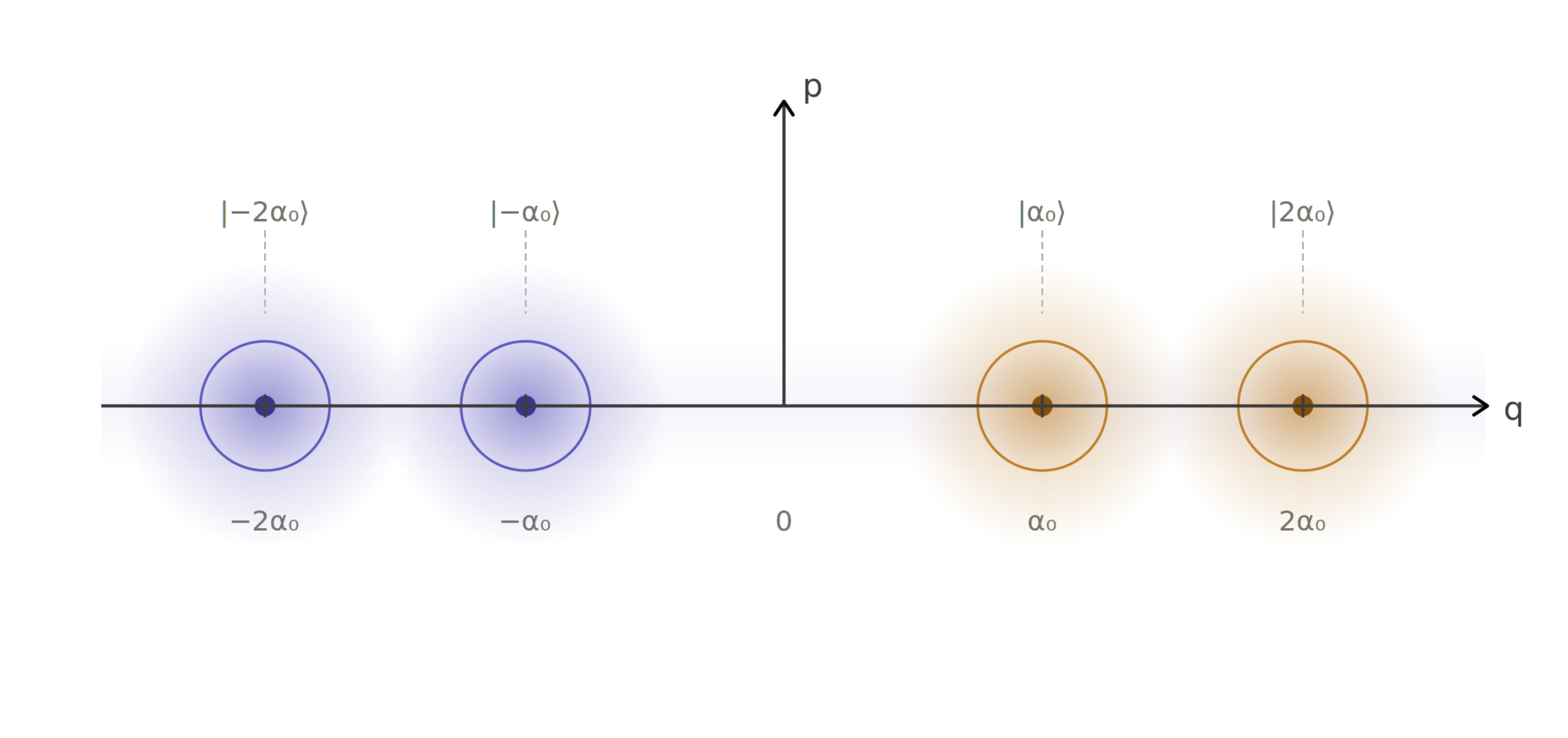}
    
    \captionof{figure}{\justifying Illustration of a unidimensional discrete modulation of coherent states along the real axis, i.e., modulation in the $\hat{q}$ quadrature, using four states.}
    \label{fig:modulation}
\end{center}
    \begin{equation}
    \gamma_{AB} :=
\begin{pmatrix}
\langle \hat{q}_A^2 \rangle_{\rho_{AB}} &
\frac{1}{2}\langle \{ \hat{q}_A, \hat{p}_A \} \rangle_{\rho_{AB}} &
\frac{1}{2}\langle \{ \hat{q}_A, \hat{q}_B \} \rangle_{\rho_{AB}} &
\frac{1}{2}\langle \{ \hat{q}_A, \hat{p}_B \} \rangle_{\rho_{AB}} \\[8pt]
\frac{1}{2}\langle \{ \hat{p}_A, \hat{q}_A \} \rangle_{\rho_{AB}} &
\langle \hat{p}_A^2 \rangle_{\rho_{AB}} &
\frac{1}{2}\langle \{ \hat{p}_A, \hat{q}_B \} \rangle_{\rho_{AB}} &
\frac{1}{2}\langle \{ \hat{p}_A, \hat{p}_B \} \rangle_{\rho_{AB}} \\[8pt]
\frac{1}{2}\langle \{ \hat{q}_B, \hat{q}_A \} \rangle_{\rho_{AB}} &
\frac{1}{2}\langle \{ \hat{q}_B, \hat{p}_A \} \rangle_{\rho_{AB}} &
\langle \hat{q}_B^2 \rangle_{\rho_{AB}} &
\frac{1}{2}\langle \{ \hat{q}_B, \hat{p}_B \} \rangle_{\rho_{AB}} \\[8pt]
\frac{1}{2}\langle \{ \hat{p}_B, \hat{q}_A \} \rangle_{\rho_{AB}} &
\frac{1}{2}\langle \{ \hat{p}_B, \hat{p}_A \} \rangle_{\rho_{AB}} &
\frac{1}{2}\langle \{ \hat{p}_B, \hat{q}_B \} \rangle_{\rho_{AB}} &
\langle \hat{p}_B^2 \rangle_{\rho_{AB}}
\end{pmatrix},
\end{equation}
can be written in the form
\begin{align}
    \gamma_{sym}&=\frac{1}{2}\left(\gamma_{AB}+(S\oplus S)\gamma_{AB} (S\oplus S)\right)\\
    \nonumber&= \begin{pmatrix}
        \langle \hat{q}_A^2 \rangle_{\rho} & 0 & \langle \hat{q}_A \hat{q}_B \rangle_{\rho} & 0 \\
        0 & \langle \hat{p}_A^2 \rangle_{\rho} & 0 & \langle \hat{p}_A \hat{p}_B \rangle_{\rho} \\
        \langle \hat{q}_A \hat{q}_B \rangle_{\rho} & 0 & \langle \hat{q}_B^2 \rangle_{\rho} & 0 \\
        0 & \langle \hat{p}_A \hat{p}_B \rangle_{\rho} & 0 & \langle \hat{p}_B^2 \rangle_{\rho}
    \end{pmatrix},
\end{align}
where $S=\text{diag}(-1,1)$. 

Given an appropriate choice of the purification $\ket{\Phi}_{AA^{'}}$, we have $\langle \hat{p}_A^2 \rangle_{\rho_{AB}} = 1$ (see Sec.\ref{sec:purification}). Therefore, $\gamma$ can be safely
replaced by $\gamma_{sym}$ when computing the secret key rate, with
\begin{align}\label{CM}
    \gamma_{sym}= \begin{pmatrix}
        V  & 0 & C_q & 0 \\
        0 & 1 & 0 & C_p \\
        C_q & 0 & W & 0 \\
        0 & C_p & 0 & W_p
    \end{pmatrix},
\end{align}
where $V=\langle \hat{q}_A^2 \rangle_{\rho_{AB}}$, $W=\langle \hat{q}_B^2 \rangle_{\rho_{AB}}$, $W_p=\langle \hat{p}_B^2 \rangle_{\rho_{AB}}$, $C_q=\langle \hat{q}_A \hat{q}_B \rangle_{\rho_{AB}}$ and $C_p=\langle \hat{p}_A \hat{p}_B \rangle_{\rho_{AB}}$. For simplicity, we denote \(\gamma_{\mathrm{sym}}\) simply as \(\gamma_{AB}\).   
As mentioned in Sec.\ref{sec: extending}, contrary to the EB scheme of the 1D GM protocol, in the discrete modulation case, it is not necessary to apply the operation of squeezing on Bob's mode. In the following, we explore the physicality region where the protocol can be implemented.

\subsection{Physicality region}
\label{physicalregion}

The parameter \( C_p \) represents the correlation between trusted modes in the \( \hat{p} \) quadrature. Since this quadrature is not modulated, this parameter is unknown. However, its value can still be bounded by the physicality requirements of quantum states, as imposed by the Heisenberg uncertainty principle inequality \cite{Serafini_2005}
\begin{align}
    \gamma_{AB} + i \Omega &\geq 0,
\intertext{or equivalently,}
\nu_- &\geq 1, \label{eq:auto}
\end{align}
where \(\nu_-\) is the smallest symplectic eigenvalue of \(\gamma_{AB}\), given by
\begin{equation}\label{eq:autosimpletico}
    \nu_- = \left(\frac{1}{2}\left(\Delta - \sqrt{\Delta^2 - 4 \det(\gamma_{AB})}\right)\right)^{1/2},
\end{equation}
with
\begin{align}\nonumber
    \Delta &= \det(\gamma_{A}) + \det(\gamma_{B}) + 2 \det(\gamma_{C}) \\\label{eq:delta}
    &= V + W W_p + 2 C_q C_p,
\end{align}
and
\begin{equation}\label{eq:det}
    \det(\gamma_{AB}) = V W W_p - V W C_p^2 + C_q^2 C_p^2 - C_q^2 W_p,
\end{equation}
where 
\begin{equation}
    \gamma_{AB}=\begin{pmatrix}
        \gamma_A & \gamma_C\\
        \gamma_C & \gamma_B
    \end{pmatrix}.
\end{equation}
From Eqs.\eqref{eq:auto} and \eqref{eq:autosimpletico} we obtain
\begin{equation}
    \Delta-\sqrt{\Delta^2-4\det(\gamma_{AB})}\geq 2,
\end{equation}
that is, 
\begin{equation}\label{eq:ineq}
    \det(\gamma_{AB})\geq \Delta-1.
\end{equation}

By substituting Eqs.\eqref{eq:delta} and \eqref{eq:det} into Eq.\eqref{eq:ineq}, we are able to obtain an inequality from which we can derive the physicality bound for the parameter \(C_p\),
\begin{equation}\label{eq:parabola}
    \left(C_p+C_0\right)^2\leq \left(1-\frac{W}{V}W_0\right)\left(W_p-W_0\right),
\end{equation}
\noindent where $C_0=\frac{C_q}{VW-C_q^2}$, $W_0=\frac{V}{VW-C_q^2}$. Appendix~\ref{physregion} presents a detailed derivation of Eq.~\eqref{eq:parabola}.

According to Eq.\eqref{eq:parabola}, the parameter \( C_p \) describes a parabola as a function of the parameter \(W_p\), and lies within the interval \([C_p^{-}, C_p^{+}]\) (see Fig. \ref{fig:physicality_region}), where
\begin{align}\label{eq:cp-} 
    C_p^{\mp}= \mp \sqrt{\left(1-\frac{W}{V}W_0\right)\left(W_p-W_0\right)}-C_0.
\end{align}








\begin{figure*}[t]
    \centering
    \includegraphics[width=0.98\textwidth]{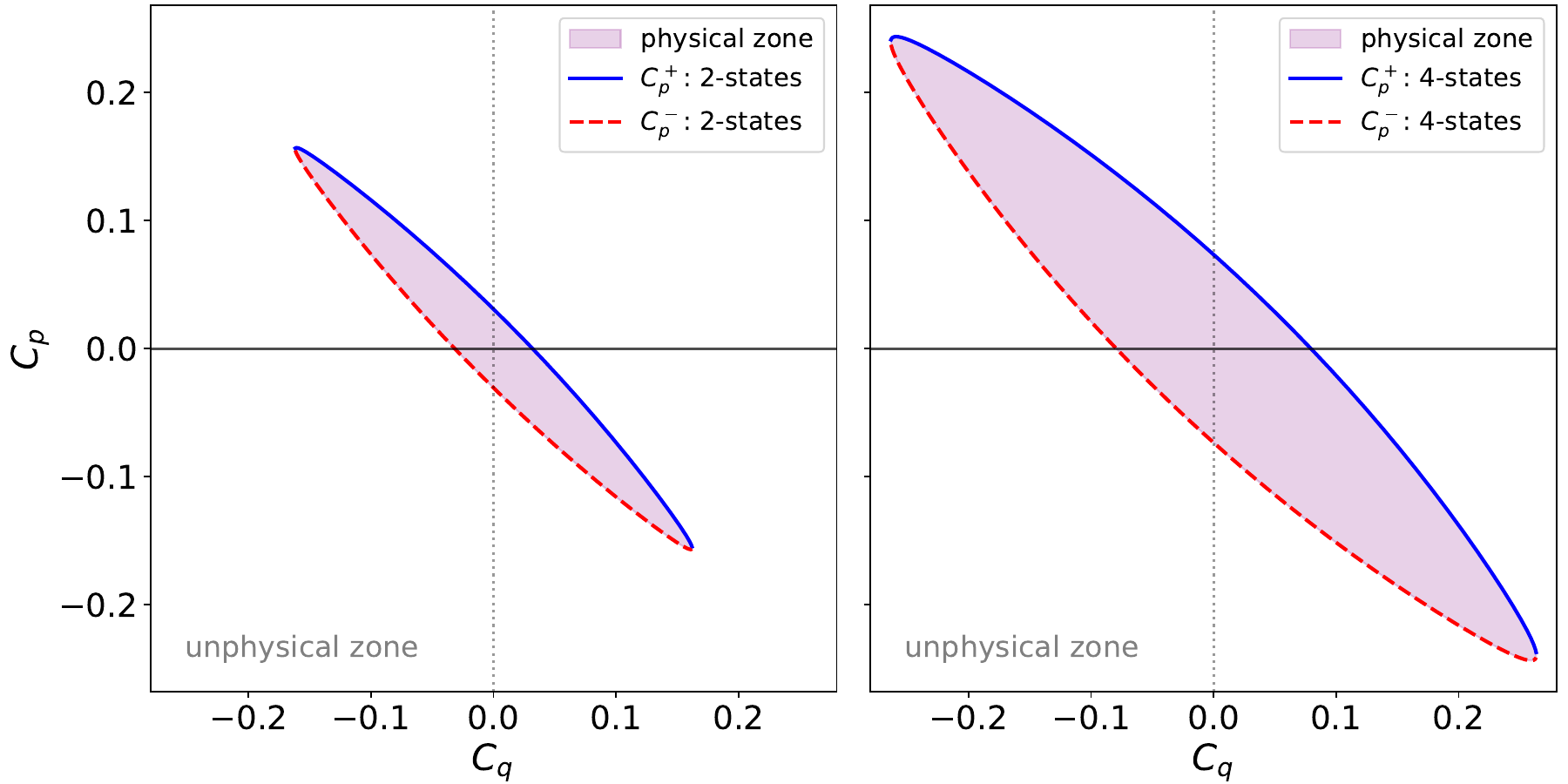}
    \caption{\justifying The physicality region comprises the purple zone and its borders, i.e., the curves for the functions \(C_p^-\) (red dashed curve) and $C_p^+$ (blue solid curve). The space beyond this region is the unphysical zone, where the protocol is not allowed to operate. Parameters: $\alpha_0 = 0.1$, $d = 10$ km, and $\xi = 0$ for 2-state (left) and 4-state (right) constellations.}
    \label{fig:physicality_region}
\end{figure*}

Let us now compute the symplectic eigenvalues of \(\gamma_{AB}\). Note that the symplectic eigenvalues of the covariance matrix are given by the square roots of the solutions to the equation
$x^2 - \Delta x + \det(\gamma_{AB}) = 0$, that is, the symplectic eigenvalues are given by
\begin{equation}
\nu_{1,2} = \sqrt{\frac{1}{2}\left( \Delta \pm \sqrt{\Delta^2 - 4 \det(\gamma_{AB})} \right)}.
\label{eigenvalues}
\end{equation}

The physicality condition implies that $\nu_{1,2} \geq 1$. At the boundary of the physicality region,  $\det(\gamma_{AB})= \Delta-1$. Therefore, $\nu_1=\sqrt{\det(\gamma_{AB})}$ and $\nu_2=1$.

The covariance matrix of Alice's state conditioned on Bob's measurement, $\gamma_{A|B}$, is given by \cite{Usenko2015,Laudenbach2018}
\begin{equation}\label{eq:conditionmatrix}
    \gamma_{A|B}=\gamma_A-\gamma_C(\Pi_q\gamma_B\Pi_q)^{-1}\gamma^{T}_C,
\end{equation}
where $\Pi_q = \mathrm{diag}(1, 0)$ selects the measured quadrature, and $M^{-1}$ denotes the pseudoinverse of the matrix $M$. We then obtain
\begin{equation}
    \gamma_{A|B}=\begin{pmatrix}
        V-W^{-1}C_q^2 & 0\\
        0 & 1
    \end{pmatrix},
\end{equation}
with symplectic eigenvalue 
\begin{equation}
    \nu_3=\sqrt{\det(\gamma_{A|B})}=\sqrt{V-\frac{C_q^2}{W}}.
\end{equation}

The Holevo information \(\chi(Y;E)\) computed for the Gaussian state with covariance matrix \(\gamma_{AB}\) is given by
\begin{equation}\label{eq:holevo}
\chi(Y;E) = g\!\left(\frac{\nu_1 - 1}{2}\right)
+ g\!\left(\frac{\nu_2 - 1}{2}\right)
- g\!\left(\frac{\nu_3 - 1}{2}\right),
\end{equation}
where $g(x) = (x + 1)\log_2(x + 1) - x\log_2(x)$.

\vspace{0.5cm}
\section{Lower bound for SKR via SDP}
\label{lowerboundsdp}

We consider the security analysis under the assumption that Eve performs a collective attack \cite{grosshans2007continuous}. In this case, Bob’s measurement process can be described by the classical-quantum (cq) state
    \begin{equation}
    \rho_{cq}=\sum_{k=0}^{N-1}\frac{p_k}{2}\left[\Pi_{2k}\otimes \mathcal{E}(|\alpha_k\rangle\langle \alpha_k|)+\Pi_{2k+1}\otimes \mathcal{E}(|-\alpha_k\rangle\langle -\alpha_k|)\right],
\end{equation}
where \(\{\Pi_i\}_{i=0}^{2N-1}\) are a set of orthonormal projectors. These projectors are associated with a classical variable that records which state was sent by Alice in each round of the protocol. 

After many rounds of the protocol, Alice and Bob each hold a sequence of random real numbers, \(X\) and \(Y\). The goal is to verify that the correlations between the strings \(X\) and \(Y\) are sufficiently strong to extract a secure key, despite Eve’s potential intervention.

\subsection{Correlation parameters}\label{correlationsection}

In our case, the correlation between Alice's and Bob's data is characterized by three parameters. The first is the correlation between Alice’s classical choices and the states received by Bob, which can be interpreted as a covariance. The remaining two parameters characterize the variances of Bob's received states in both the $\hat{q}$ and $\hat{p}$ quadratures.

To extract the values of these parameters, we adapt to our protocol the expressions derived by Ghorai \textit{et al.} \cite{Ghorai2019}, obtaining
    \begin{equation}
        c_2=\mathrm{Tr}\left[\left((\Pi_0-\Pi_1)\otimes\hat{q}\right)\rho_{cq}\right],
    \end{equation}
    \begin{equation}\label{eq:vx}
    v_q=\mathrm{Tr}\left[\left(\mathds{1}\otimes\hat{q}^2\right)\rho_{cq}\right],
    \end{equation}
    \begin{equation}\label{eq:vq}  v_p=\mathrm{Tr}\left[\left(\mathds{1}\otimes\hat{p}^2\right)\rho_{cq}\right],
    \end{equation}
for the BPSK case. In the general case with \(2N\) states, the parameters \(v_q\) and \(v_p\) are defined in the same way. On the other hand, we can generalize the expression for the parameter \(c_2\), obtaining
    \begin{equation}
        c_{2N}=\mathrm{Tr}\left[\left(\left(\sum_{k=0}^{N-1}\Pi_{2k}-\sum_{k=0}^{N-1}\Pi_{2k+1}\right)\otimes \hat{q}\right)\rho_{cq}\right].
    \end{equation}

For instance, consider a bosonic, phase-sensitive Gaussian channel characterized by transmittances \(T_q\) and \(T_p\), and excess noises \(\xi_q\) and \(\xi_p\) for the quadratures \(\hat{q}\) and \(\hat{p}\), respectively. In this case, for BPSK modulation, the parameters \(c_2\), \(v_q\), and \(v_p\) can be expressed in terms of the channel parameters as follows: 
\begin{align*}
    c_2(T_q,\xi_q)&=2\sqrt{T_q}\alpha_0,\\
    v_q(T_q,\xi_q)&=T_q(V-1)+1+T_q\xi_q,
\end{align*} 
\noindent with $V=4\alpha_0^2+1$ and 
$$v_p(T_p,\xi_p)=1+T_p\xi_p.$$

For a general modulation, the parameters \(c_{2N}\) and \(V\) change. We identify the recursive relationship with respect to the state amplitudes and probability weights in the modulation scheme as
$$c_{2N}(T_q,\xi_q)=2\sqrt{T_q}\alpha_0\sum_{k=0}^{N-1}p_kr_k,$$
and
$$V=4\alpha_0^2\left(\sum_{k=0}^{N-1}p_kr_k^2\right)+1.$$

\subsection{Purification and SDP definition}
\label{sec:purification}
Considering the average state of the constellation \(\tau = \sum_{k=0}^{N-1}p_k/2(|\alpha_k\rangle\langle\alpha_k|+|-\alpha_k\rangle\langle-\alpha_k|)\), a purification of \(\tau\) is prepared by Alice in the equivalent EB protocol, and one part of it is sent to Bob through the channel.

Alice is free to choose this purification independently of Eve’s intervention, since the equivalence between the P\&M and EB protocols does not depend on it \cite{lin2019asymptotic}. Therefore, we select the purification that maximizes the correlation between Alice and Bob within the Gaussian approximation method. This choice consequently minimizes the Holevo bound on Eve’s accessible information and yields the tightest security bounds \cite{Ghorai2019,Denys2021}. In this sense, an appropriate choice is given by the state
\begin{align}
    |\Phi\rangle&=(\mathbb{I}\otimes \tau^{1/2})\sum_{n=0}^{\infty}|n,n\rangle\\
   \nonumber &=\sum_{k=0}^{N-1}\sqrt{\frac{p_k}{2}}(|\psi_{2k}\rangle |\alpha_k\rangle+|\psi_{2k+1}\rangle |-\alpha_k\rangle),
\end{align}
where the states $\ket{\psi_i}$ are defined as in \cite{Denys2021}
\begin{align}
    \ket{\psi_{2k}}=\sqrt{\frac{p_k}{2}}\tau^{-\frac{1}{2}}\ket{\alpha_k},
\end{align}
and
\begin{equation}
    \ket{\psi_{2k+1}}=\sqrt{\frac{p_k}{2}}\tau^{-\frac{1}{2}}\ket{-\alpha_k}.
\end{equation}

Since the channel can be represented using Kraus operators \(\{E_j\}\), which satisfy $\sum_j E_j^{\dagger}E_j=\mathds{1}$ \cite{nielsen2010quantum}, we have that the state shared by Alice and Bob in the EB protocol is given by
\begin{equation}
    \rho_{AB}=\sum_{k,l=0}^{2N+1}\frac{\sqrt{p_kp_l}}{2}|\psi_k\rangle\langle\psi_l|\otimes\left(\sum_j E_j|\alpha_k\rangle\langle\alpha_l|E_j^{\dagger}\right),
\end{equation}
where $|\alpha_{2i}\rangle=|\alpha_{i}\rangle$, $|\alpha_{2i+1}\rangle=|-\alpha_{i}\rangle$ and $p_i=p_{i+1}$ for $i=0,1,\ldots,2N-1$. For this purification, we have \( \rho_A = \mathrm{Tr}_B(\rho_{AB}) = \tau \), and thus
\begin{align*}
    \langle \hat{p}_A^2\rangle_{\rho_{AB}}=\mathrm{Tr}(\hat{p}_A^2 \rho_{AB})=\mathrm{Tr}(\hat{p}_A^2\tau)=\sum_{k=0}^{N-1}2\frac{p_k}{2}=1.
\end{align*}

By analyzing the covariance matrix in Eq.\eqref{CM}, we are interested in five parameters: the variance of \(\rho_A\) in the $\hat{q}$ quadrature ($V_A$), the variances of \(\rho_B=\mathrm{Tr}_A(\rho_{AB})\) in the \(\hat{q}\) and \(\hat{p}\) quadratures ($W$ and $W_p$, respectively), the covariance in the modulated quadrature ($C_q$), and the covariance in the unmodulated quadrature ($C_p$). The parameter $V$ depends only on the modulation format and is independent of the quantum channel between Alice and Bob, and Bob can measure $W=v_q$ and $W_p=v_p$ locally. 

Therefore, the only unknown parameters remaining in the covariance matrix (Eq.\eqref{CM}) are \(C_q\) and $C_p$. This implies that the Holevo information calculated from Eq.\eqref{eq:holevo} is a function of \(C_q\) and \(C_p\), 
that is, \(\chi(Y;E) = f(C_q, C_p)\). 

The parameter \(C_p\) is unknown, but must be constrained by the physicality region given by Eq.\eqref{eq:parabola}. Therefore, for each value of \(C_q\), we can take \(C_p \in [C_p^-, C_p^+]\) that maximizes the Holevo information, thus avoiding underestimation of Eve’s ability to obtain information. 
That is, for each \(C_q\), we choose \(C_p^{\mathrm{max}} \in [C_p^-, C_p^+]\) such that
\[
    f(C_q, C_p^{\mathrm{max}}) = \max_{C_p \in [C_p^-, C_p^+]} f(C_q, C_p).
\]

With this pessimistic choice for the parameter \(C_p\), the Holevo information becomes a function that depends only on the parameter \(C_q\). 
One can check numerically that the function \(C_q \mapsto f(C_q, C_p^{\mathrm{max}})\) is decreasing for \(C_q \ge 0\). An analytical proof of this property is technically challenging and is not currently accessible in closed form; therefore, we restrict ourselves to its numerical verification (see Appendix \ref{app:monotonicity}). It is important to emphasize that the proposed security proof relies on this behavior of the Holevo information; however, it can be easily verified as an additional step within the computational algorithm used to evaluate the key rate.
In what follows, we show another property that fully characterizes the behavior of the function \(f(C_q, C_p)\). 
\begin{proposition}\label{prop:even}
    Within the physicality region, \(f(C_q, C_p^{\mathrm{max}})\) is an even function, that is,
    \begin{equation}
        f(C_q, C_p^{\mathrm{max}}) = f(-C_q, C_p^{\mathrm{max}}).
    \end{equation}
\end{proposition}
\begin{proof}
    For all $C_q$ and $C_p$ in the physicality region (see Fig.~\ref{fig:physicality_region}), we have\footnote{The physicality region is symmetric in the parameters $C_q$ and $C_p$.}   
\[
\nu_i(C_q,C_p)=\nu_i(-C_q,-C_p), \quad \text{for } i=1,2,3.
\]
Therefore, $f(C_q,C_p)=f(-C_q,-C_p)$, and consequently
\begin{align*}
f(-C_q,C_p^{\mathrm{max}})
&=\max_{C_p} f(-C_q,C_p)\\
&=\max_{C_p} f(C_q,-C_p)\\
&=f(C_q,C_p^{\mathrm{max}}).
\end{align*}
\end{proof}
Therefore, \(f(C_q, C_p^{\mathrm{max}})\) is symmetric around \(C_q = 0\), meaning that it increases for \(C_q \le 0\) (see Fig.~\ref{fig:Holevo maximum}). Hence, the maximum Holevo information is a decreasing function of \(|C_q|\).

\begin{figure}[h!]
    \centering
    \includegraphics[width=0.8\textwidth]{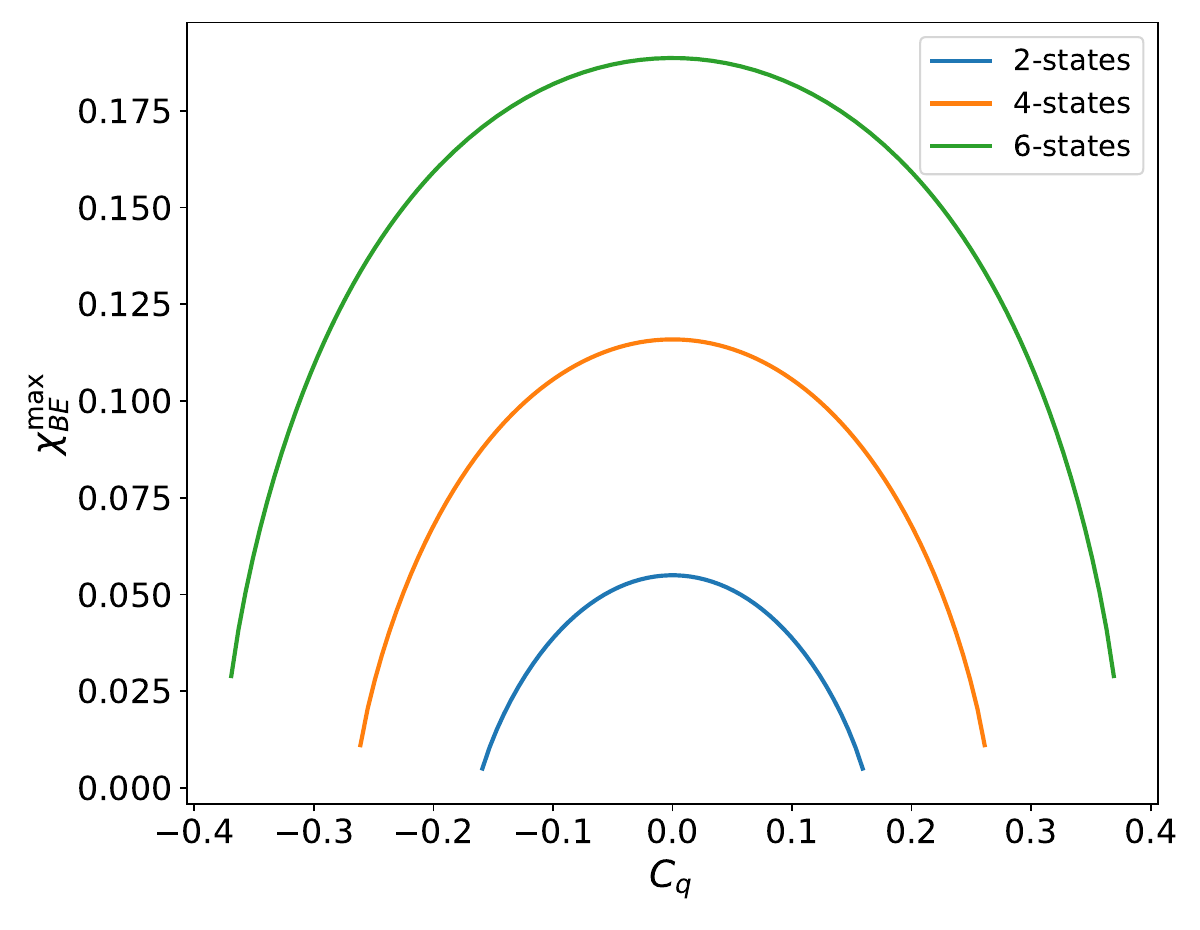}
    \caption{\justifying Maximum Holevo information as a function of the correlation \(C_q\) with \(\alpha_0 = 0.1\), \(d = 10\) km, \(\xi = 0\) for 2 (blue bottom curve), 4 (orange middle curve), and 6-state (green upper curve) constellations.}
    \label{fig:Holevo maximum}
\end{figure}

The parameter \(C_q\) is defined as the expectation value of \(\hat{q}_A \hat{q}_B\) for the state \(\rho_{AB}\), that is,
\begin{equation}
    C_q = \mathrm{Tr}\left[(\hat{q}_A \hat{q}_B) \, \rho_{AB}\right].
\end{equation}
We can obtain an upper bound on the Holevo information (or, equivalently, a lower bound on the secure key rate) 
compatible with any quantum channel whose shared state between Alice and Bob reproduces the fixed values of 
\(c_{2N}\), \(v_q\), and \(v_p\), by determining the minimal value of \(|C_q|\) as a function of these parameters.

As in the analysis presented in Ref.\cite{Ghorai2019}, we consider the orthogonal projector onto the subspace spanned by the coherent states $\{|\alpha_k\rangle, |-\alpha_k\rangle\}_{k=0}^{N-1}$, given by \(P = \sum_{k=0}^{2N-1}|\psi_k\rangle\langle\psi_k|\).

To bound the key rate $K$ in Eq.\eqref{key rate}, we must then find the value of
\begin{equation}
    C_q(\rho) = \mathrm{Tr}\left[ (P \hat{q}_A P \otimes \hat{q}_B) \rho \right]
\end{equation}
that maximizes the Holevo information, under the constraints that \(\rho \succeq 0\), \(\mathrm{Tr}(\rho) = 1\), and
\[
\mathrm{Tr}(B_0 \rho) = v_q, \quad \mathrm{Tr}(B_1 \rho) = v_p, \quad \mathrm{Tr}(B_2 \rho) = c_{2N},
\]
where $B_0=P\otimes \hat{q}^2$, $B_1=P\otimes \hat{p}^2$ and $B_2=(\sum_{k=0}^{N-1}|\psi_{2k}\rangle\langle\psi_{2k}| -|\psi_{2k+1}\rangle\langle\psi_{2k+1}|)\otimes\hat{q}$. Another condition is that the partial trace over Bob's subsystem satisfies
\begin{align}
    \mathrm{Tr}_B (\rho) = \mathrm{Tr}_B \left( |\Phi\rangle\langle\Phi| \right) =\tau.
\end{align}

Therefore, we are interested in the following optimization problem:
\begin{equation}
\begin{split}
\text{minimize:} \quad & |C_q(\rho)|\\ 
\text{subject to:} \quad
& \mathrm{Tr}_B(\rho) = \mathrm{Tr}_B \left( |\Phi\rangle\langle\Phi| \right)=\tau,\\
& \mathrm{Tr}(B_0\rho) = v_q,\\ 
& \mathrm{Tr}(B_1\rho) = v_p,\\ 
& \mathrm{Tr}(B_2\rho) = c_{2N}, \\
& \rho \succeq 0.
\label{sdp}
\end{split}
\end{equation}

We extract the value $C_q^{*}$ such that $|C_q^{*}|$ is the solution of the optimization problem. Thus, $C_q^{*}$ represents the optimal value of $C_q$ that maximizes the Holevo quantity, that is,
\[
f(C_q^{*}, C_p^{\mathrm{max}}) = \max_{C_q} f(C_q, C_p^{\mathrm{max}}).
\]

Moreover, knowing \( C_q^* \), we can define an ``optimal'' value of \( C_p \) as
\begin{align}
    C_p^* =
\begin{cases}
0, & \text{if } 0 \in [C_p^-, C_p^+]^*,\\ 
C_p^{-*}, & \text{if } C_p^{-*}> 0,\\
C_p^{+*}, & \text{if } C_p^{+*}<0,
\end{cases}
\label{optimalcp}
\end{align}
where \([C_p^-, C_p^+]^*\), \(C_p^{-*}\), and \(C_p^{+*}\) denote the values of \([C_p^-, C_p^+]\), \(C_p^{-}\), and \(C_p^{+}\) corresponding to \(C_q^*\), respectively. The choice of these optimal values for \(C_p\) is made by analyzing the behavior of the Holevo information within the physicality region.

We observe that, analogously to the choice made by Usenko and Grosshans in Ref.\cite{Usenko2015}, a suitable choice for the most pessimistic value of \(C_p\) corresponds to the largest possible negative value when the region allows negative values and to the smallest positive value otherwise. The values of \(C_q\) at which these changes in behavior occur are given by
\begin{equation}
    C_{q_\pm}=\pm\sqrt{\frac{VWW_p-V-WW_p+1}{W_p}}. 
\end{equation}

For \( C_q^* > C_{q_+} \) or \( C_q^*< C_{q_-} \), the pessimistic value (which maximizes Holevo information) \( C_p^{\mathrm{max}} \) is close to the optimal value \( C_p^* \), so \( C_p^* \) can be used. Moreover, this optimal value lies on the boundary of the physicality region, allowing the corresponding symplectic eigenvalues to be computed analytically.

For \( C_{q_-} \leq C_q^* \leq C_{q_+} \), we also have \( C_p^* \approx C_p^{\mathrm{max}} \); however, the symplectic eigenvalues must be determined numerically.

In this regard, we can compute an upper bound on \(\sup \chi(Y;E)\) by evaluating the Holevo information for the Gaussian state \(\rho^*\) with covariance matrix
\begin{equation}\label{CM: optimized}
    \gamma^* =
    \begin{pmatrix}
        V & 0 & C_q^* & 0\\ 
        0 & 1 & 0 & C^*_p \\
        C_q^* & 0 & v_q & 0 \\
        0 & C^*_p & 0 & v_p
    \end{pmatrix}.
\end{equation}

The extremality property of Gaussian states \cite{wolf2006} and the construction of the optimization problem ensure that
\begin{equation}\label{eq:limit}
    \chi(Y;E)_{\rho^*} \geq \sup \chi(Y;E),
\end{equation}
where the supremum is taken over all channels compatible with the parameters \(c_{2N}\), \(v_q\), and \(v_p\).
\begin{remark} 
    The upper bound in Eq.~\eqref{eq:limit} is certified by the solution of the optimization problem, provided that the numerical verification procedure described in Appendix~\ref{app:Numerical_verificatio} is successfully satisfied.
\end{remark}
\section{Numerical results}
\label{numerical}

In this section, we present numerical results for the optimization problem given by Eq.~\eqref{sdp}. We compare our findings with known results, such as for the pure-loss and Gaussian phase-insensitive channels, including the 1D GM protocol~\cite{Usenko2015}, with the same transmittance and noise in both quadratures. It is worth mentioning that our method is not restricted to a specific channel model; it is widely applicable to any channel, for instance, the phase-sensitive channel, where the transmittance and noise can be different in each quadrature. In the following, the SDP was implemented in CVXPY (Python) \cite{diamond2016cvxpy} and we used the MOSEK solver \cite{mosek_python_api_2025}. More details are given in Appendix \ref{app:methodology}.

In order to compute the first term of the SKR in Eq.(\ref{key rate}), we use the classical mutual information for the 1D Gaussian case, defined in Ref.\cite{Usenko2015} as
\begin{equation}
    I(X;Y)=\frac{1}{2}\log\left(\frac{V_A}{V_{A|B}}\right),
\end{equation}
while in the discrete-modulated case we have $V_A=V$ and $V_{A|B}=V-C_q^2W^{-1}$. Hence,
\begin{equation}\label{eq: mutual}
    I(X;Y)=\frac{1}{2}\log\left(\frac{V}{V-C_q^2W^{-1}}\right).
\end{equation}

Since the channel is unknown in our problem, the mutual information also depends on the solution to the SDP. In our case, the $C_q$ term in Eq.(\ref{eq: mutual}) is the optimal value $C_q^*$ given by the optimization. In Fig. \ref{mutualholevoinfopure}, we show (a) the mutual and (b) Holevo information, comparing the pure-loss channel (solid curves) and our optimization problem (dashed curves) for 1D constellations with 2, 4, and 6-states, according to $r_k=k+1$ for $\alpha_0 = 0.1$ and uniform probability distribution, yielding variances $V = \{1.04, 1.10, 1.18\}$, respectively. The Gaussian mutual information is a good approximation for the discrete-modulated approach when the variance of the modulated quadrature is low \cite{Wu2010}, which can be observed in Fig. \ref{fig:mutualinfo}. On the other hand, as shown in Fig. \ref{fig:holevoinfo}, the Gaussian extremality assumption significantly overestimates the Holevo information compared to the pure-loss channel, with the overestimation growing with increasing constellation size. For constellations with more than four states, the SKR decreases instead of approaching the Gaussian limit, showing that bounds derived from Gaussian extremality become increasingly loose as the constellation grows. Our results explicitly demonstrate the limitations of this widely used approximation and underscore the necessity of alternative methods, such as the numerical approach of Lin \textit{et al.} \cite{lin2019asymptotic}, for accurate security analysis of 1D DM protocols, as demonstrated in \cite{wu2026discretely}.

\begin{figure}[h]
    \centering

    \subfloat[Mutual information\label{fig:mutualinfo}]{
        \includegraphics[width=0.6\textwidth]{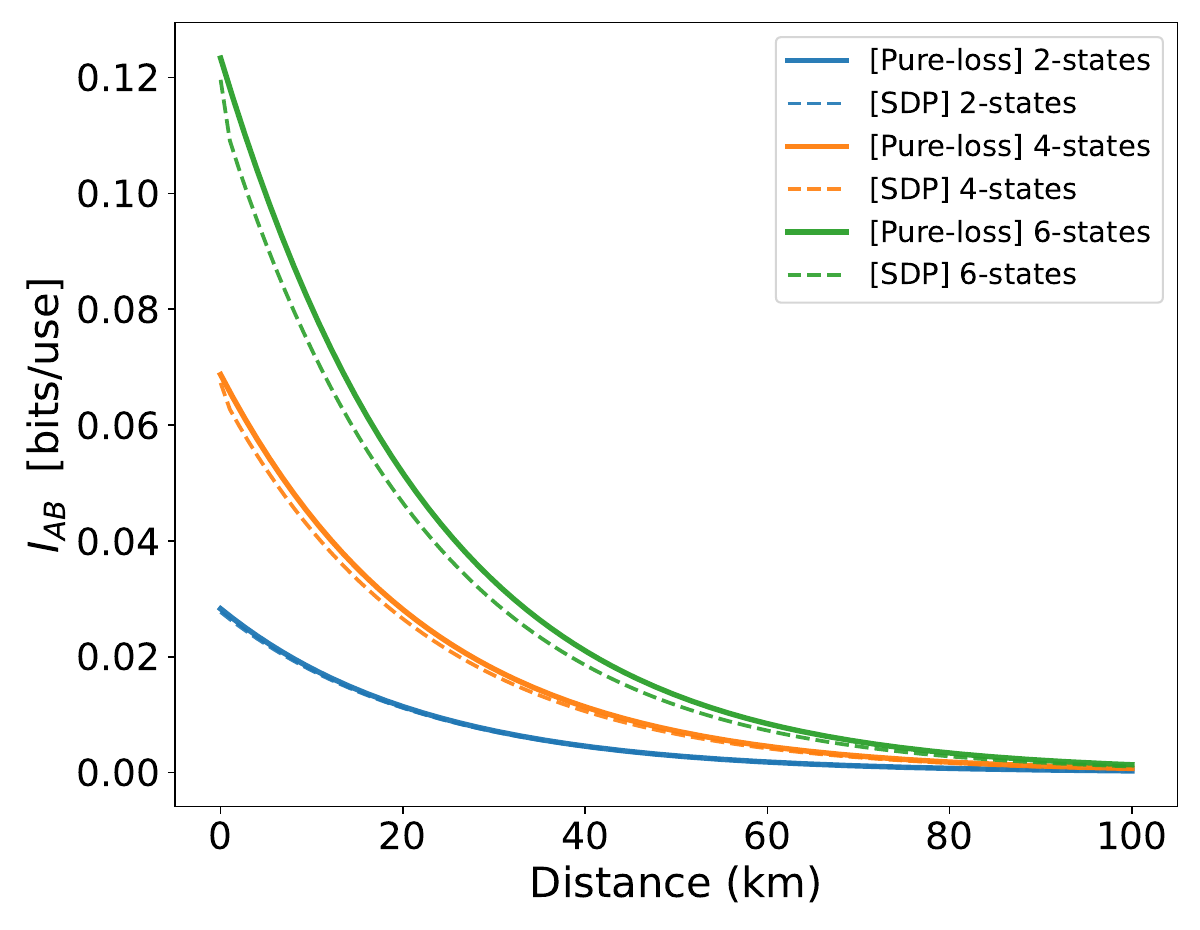}
    }

    \vspace{1.5ex}

    \subfloat[Holevo information\label{fig:holevoinfo}]{
        \includegraphics[width=0.6\textwidth]{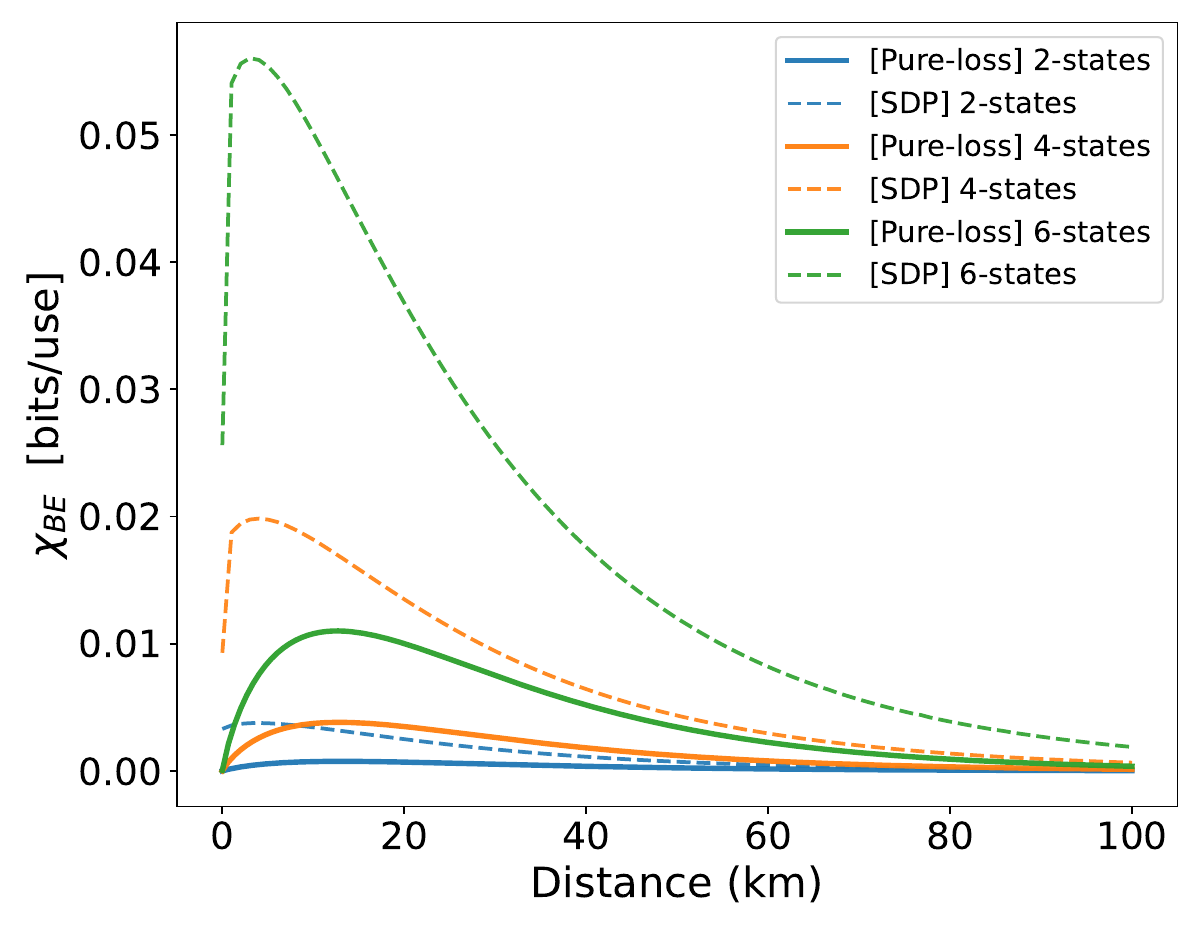}
    }

    \caption{\justifying
    (a) Mutual and (b) Holevo information for the 2-, 4-, and 6-state constellations
    for the pure-loss (solid curves) and SDP (dashed curves) models with variance of $V = \{1.04, 1.10, 1.18\}$, respectively.
    Parameters: $\alpha_0=0.1$, $\xi=0$.
    }

    \label{mutualholevoinfopure}
\end{figure}

In this limited scenario, we show in Fig. \ref{fig:skrcombined} a comparison between the 1D Gaussian, pure-loss, and SDP cases for the 2- and 4-state constellations under the same variance for each constellation, with $\alpha_0 = 0.1$ and $r_1=2$. To focus on the impact of the Gaussian extremality assumption, we assume perfect reconciliation efficiency ($\beta = 1$). The solid and dashed curves represent the Gaussian and pure-loss cases, respectively, while the dotted and dash-dotted curves represent the SDP case with and without the correlation estimation in $\hat p$ quadrature, respectively. This means that, with correlation estimation, the $C_p^*$ value in $[C_p^-,C_p^+]$ maximizes Eve's information. Thus, the security is more conservative: $\hat{p}$ quadrature is unmodulated and uncontrolled, so Eve may exploit the unknown correlation $C_p$, yielding a worst-case scenario for Alice and Bob, with a higher Holevo information and lower secret key rate. In contrast, without estimation, $C_p^*$ is set to the boundary of $[C_p^-,C_p^+]$, with the smallest absolute value. This assumption is more optimistic, yielding a lower Holevo information and higher secret key rate. We observe that our method provides a good approximation to the 2-state case, showing strong agreement with both the numerical pure-loss case and the Gaussian limit. However, for a 4-state constellation, the Gaussian approximation fails significantly, with the protocol's performance falling below that of the 2-state scheme.

\begin{figure}[h!]
    \centering
    \includegraphics[width=0.6\textwidth]{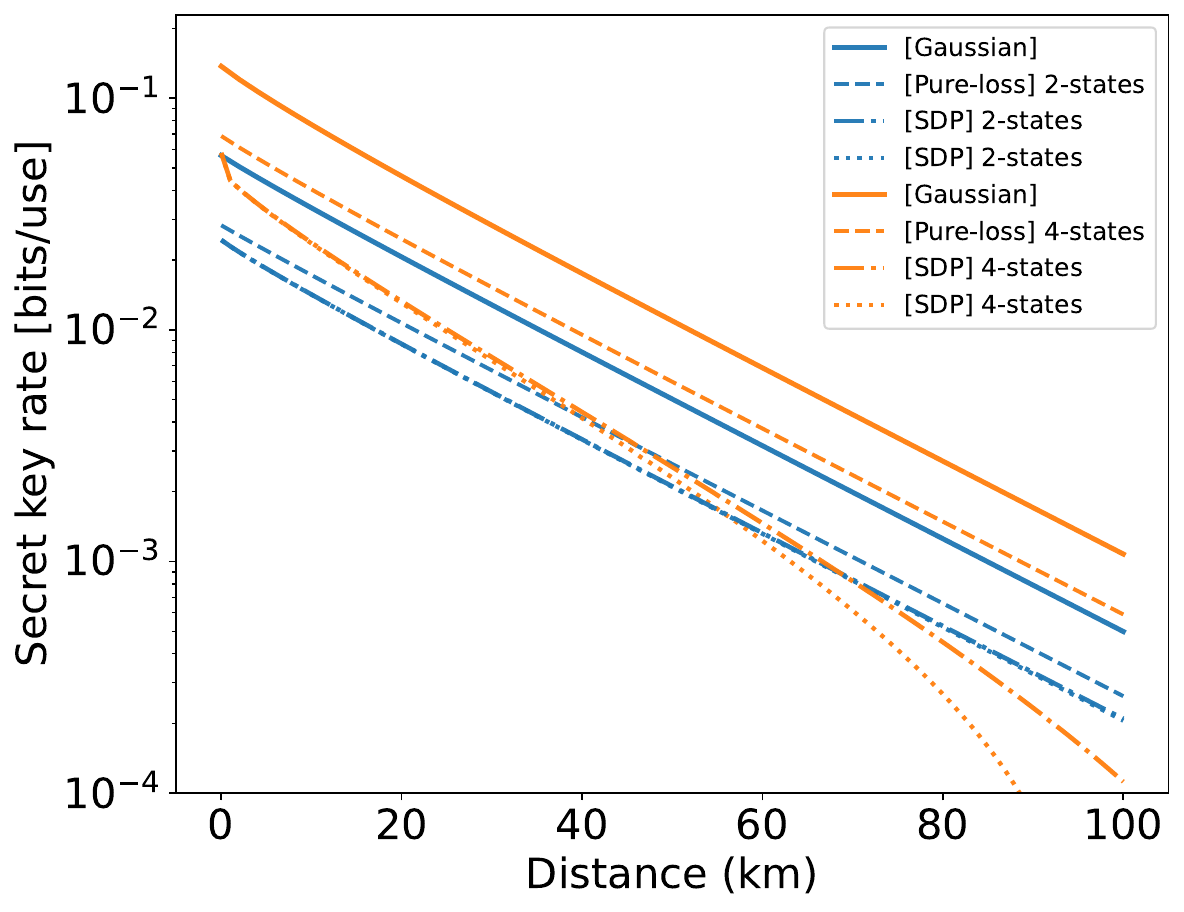}
    \caption{\justifying Comparison between the 1D Gaussian-modulated CV-QKD protocol (solid curves) with the pure-loss (dashed curves) and SDP models for the discrete-modulated protocol with 2- and 4-state constellations. The dotted (dash-dotted) curve represents the SDP model with (without) correlation estimation in \(p\). The variance is \(V = 1.04 (1.1)\) for the 2(4)-state cases and their respective Gaussian curves. Parameters: \(\alpha_0 = 0.1\), \(\xi = 0\), and $\beta = 1$. For this distance range, the dotted and dash-dotted blue curves are overlapping.}
    \label{fig:skrcombined}
\end{figure}

\begin{figure}[h!]
    \centering
    \includegraphics[width=0.6\textwidth]{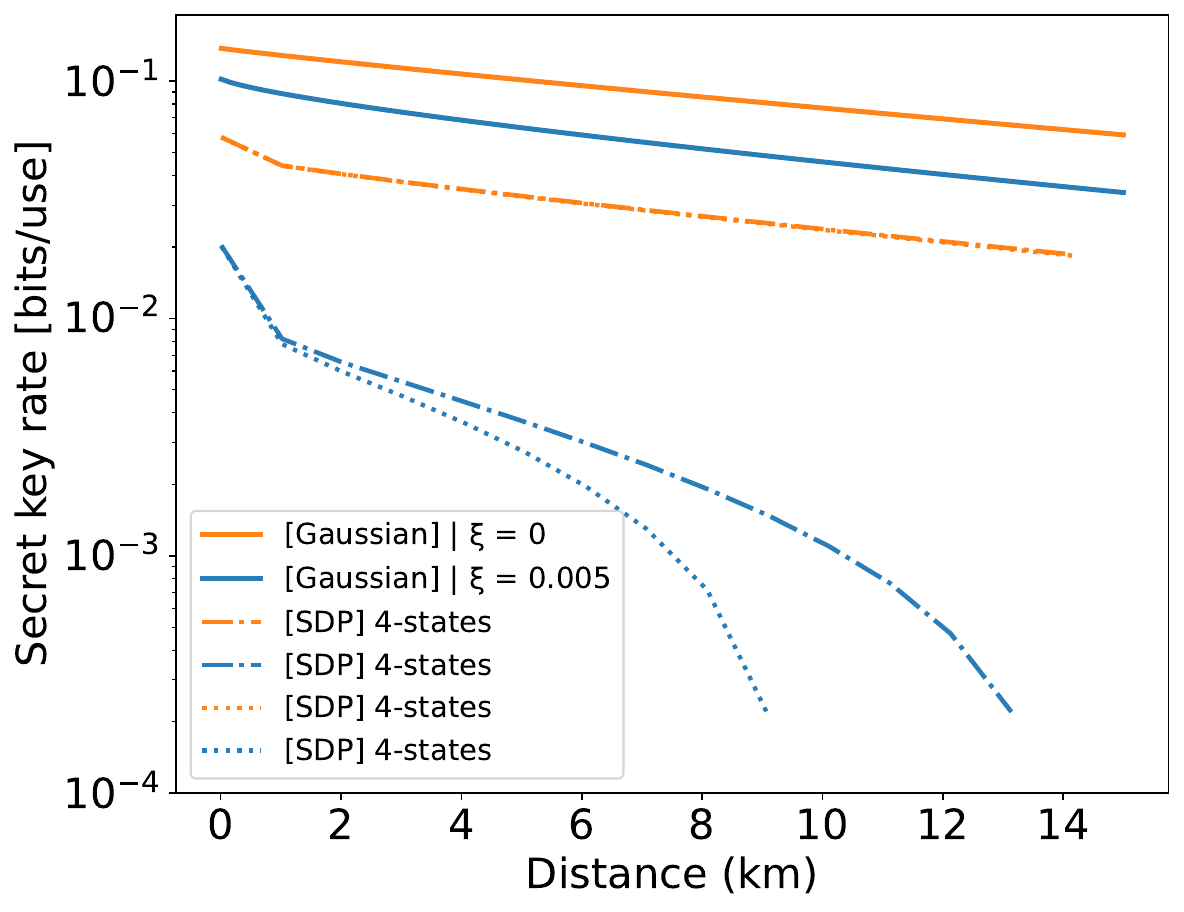}
    \caption{\justifying Comparison between the 1D Gaussian-modulated CV-QKD protocol (solid curves) with the SDP model for the discrete-modulated protocol with 4-state constellation for \(\xi = 0\) (orange curves) and \(\xi = 0.005\) (blue curves) with $\beta = 1$. The dotted (dash-dotted) curve represents the SDP model with (without) correlation estimation in \(p\). The variance is \(V = 1.04 (1.1)\) for the 2(4)-state cases and their respective Gaussian curves. For this distance range, the dotted and dash-dotted orange curves are overlapping.}
    \label{fig:skrcombined2}
\end{figure}

Additionally, we analyze the protocol's performance in the presence of excess noise. As an example, we obtain a non-zero SKR for the 4-state protocol for $\xi = 0.005$. Again, the solid curves are the 1D GM protocol for $\xi = 0$ (orange) and $\xi = 0.005$ (blue), and the dotted (dash-dotted) curves represent the SDP solution with (without) the correlation estimation in the $\hat{p}$ quadrature. As expected, no secret key can be extracted for a 2-state constellation due to its limited robustness, in agreement with the findings of Ref.\cite{Zhao2009}. However, while one might expect a positive SKR for larger constellations, the overestimation shown in Fig. \ref{fig:holevoinfo} is even worse when some excess noise is taken into account, preventing any SKR extraction.

Although the results presented above correspond to a single value of $\alpha_0$, Fig. \ref{fig:skralpha} reveals that the Gaussian extremality assumption severely constrains the viable range of modulation amplitudes for any given distance. Increasing the number of states improves the SKR only for amplitudes close to the vacuum limit, indicating that the approximation becomes progressively worse for larger constellations and modulation amplitudes. We tested constellation sizes up to $2N=12$ and found that feasible key rates are achievable only for up to $2N=8$ for increasingly more restrictive parameter regimes (near-vacuum modulation, very low excess noise), as shown in Figs.~\ref{fig:skralpha} and \ref{fig:skralpha2}. This limitation worsens in the presence of excess noise (see Fig. \ref{fig:skralpha2}), confirming that the Gaussian extremality approach is limited to small constellations for this class of protocols.

\begin{figure}[h!]
    \centering

    \subfloat[\label{fig:skralpha}]{
        \includegraphics[width=0.6\textwidth]{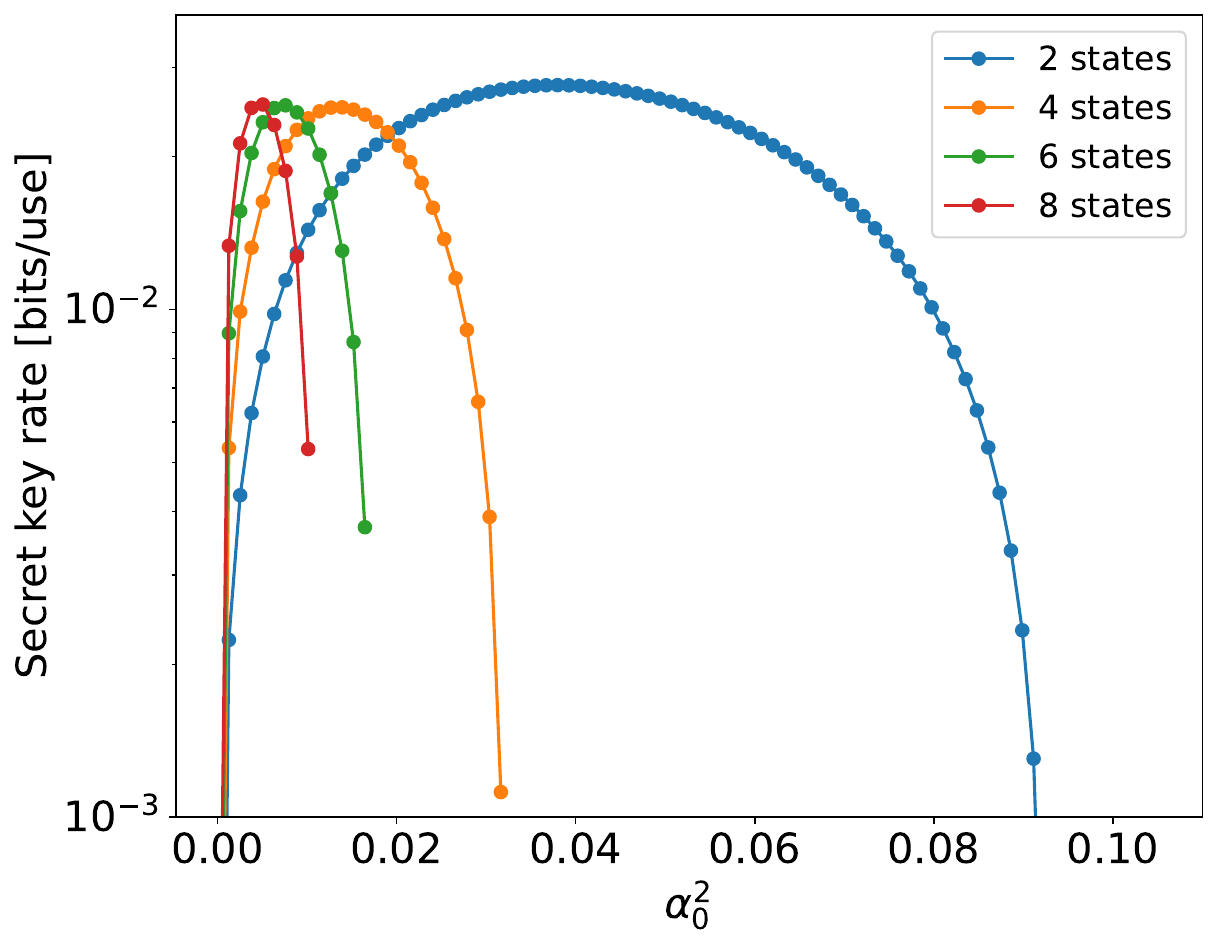}
    }

    \vspace{1.5ex}

    \subfloat[\label{fig:skralpha2}]{
        \includegraphics[width=0.6\textwidth]{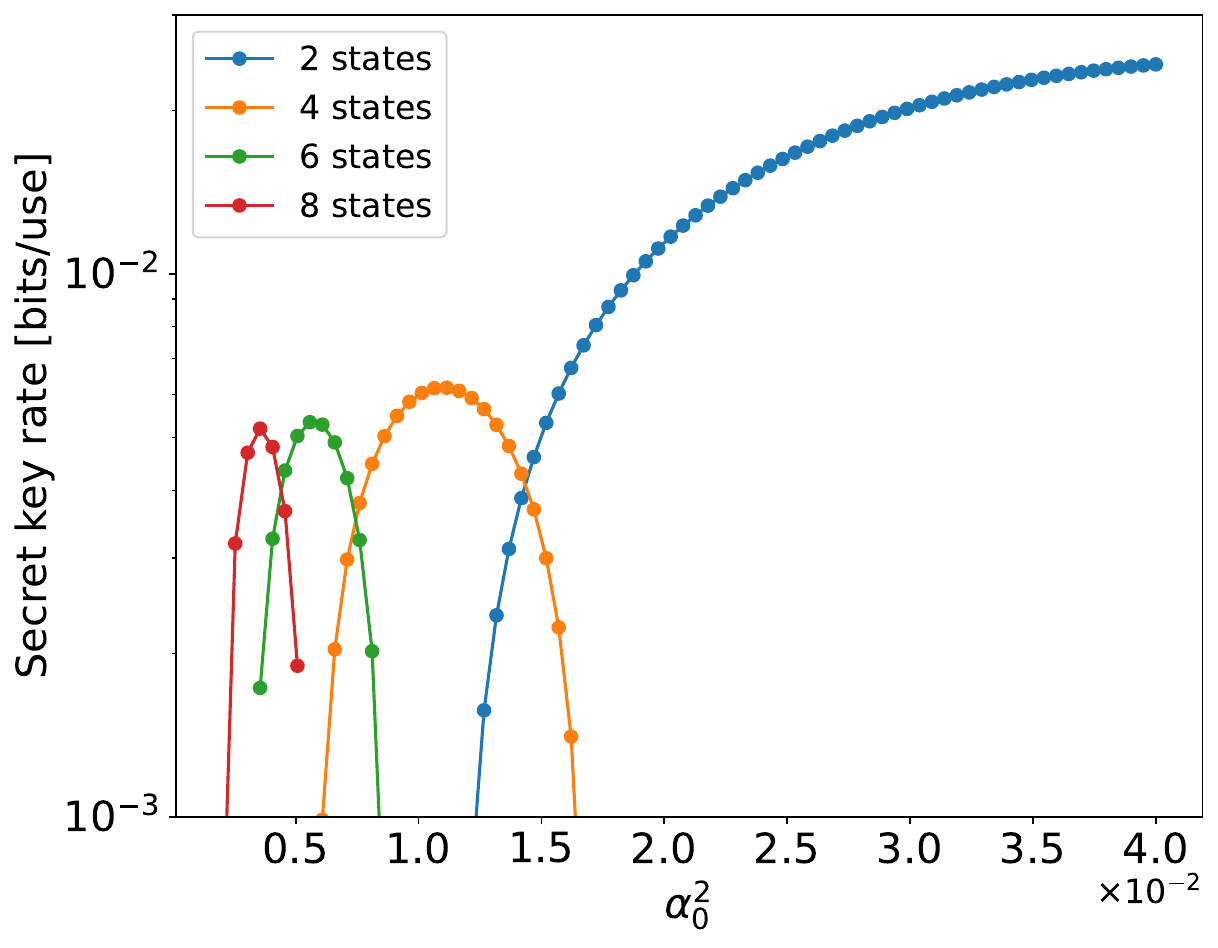}
    }

    \caption{\justifying Secret key rate as a function of $\alpha_0^2$ for the 2, 4, 6, and 8-state protocols for (a) $d = 10$ km, \(\xi = 0\), (b) $d = 2$ km, \(\xi = 0.005\), and $\beta = 1$. In the low-energy regime, the states are modulated near the origin of phase space, so that the measurement statistics do not exhibit a privileged direction. As a consequence, the effective distribution becomes approximately isotropic and closer to a Gaussian distribution~\cite{Leverrier2011}. This approximation to the Gaussian regime makes the method based on the Gaussian extremality more efficient in this scenario. 
    }

    \label{mutualholevoinfo}
\end{figure}

This behavior is counterintuitive and contrasts with previous results reported for 2D DM CV-QKD protocols, in which increasing the constellation size typically improves performance, for example, Ref.\cite{Denys2021}. We consistently observed this trend at different distances by optimizing the modulation variance for each constellation, as shown in Figs.~\ref{fig:skralphaoptimized} and \ref{fig:skralpha2optimized} (in the pure-loss scenario and in the presence of excess noise, respectively). This confirms that increasing the modulation variance leads to an overestimation of the Holevo information and therefore to lower SKRs, while at low variances, the SKR saturates for constellations with a small number of states (in most cases, for two states).

We attempted to improve the performance through probabilistic shaping by assigning a discrete Gaussian probability distribution to uniformly distributed coherent states along the real axis. However, this modification did not yield any improvement, and choosing a probability distribution close to uniform turns out to yield better performance (see Appendix \ref{app:methodology}).

\begin{figure}[h!]
    \centering

    \subfloat[\label{fig:skralphaoptimized}]{
        \includegraphics[width=0.6\textwidth]{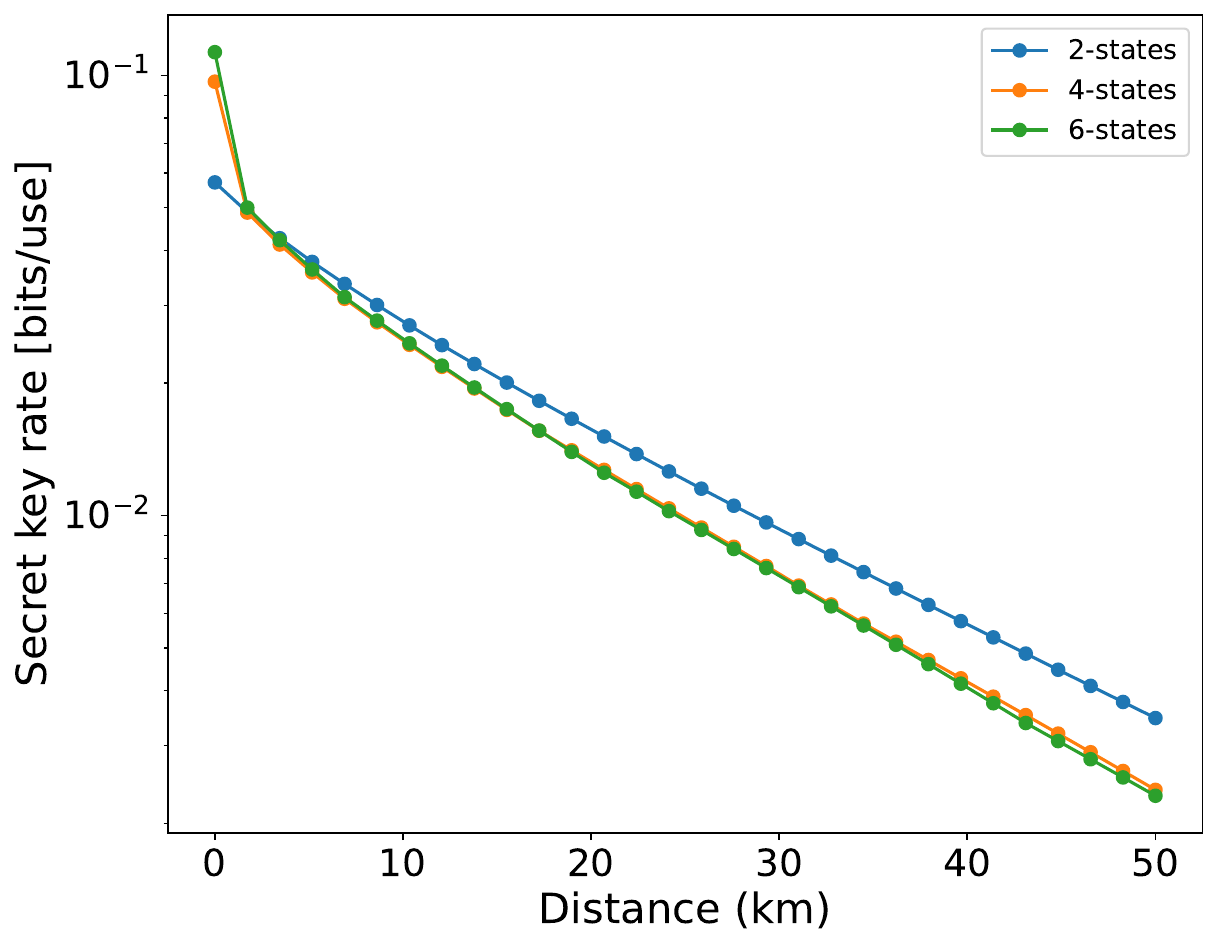}
    }

    \vspace{1.5ex}

    \subfloat[\label{fig:skralpha2optimized}]{
        \includegraphics[width=0.6\textwidth]{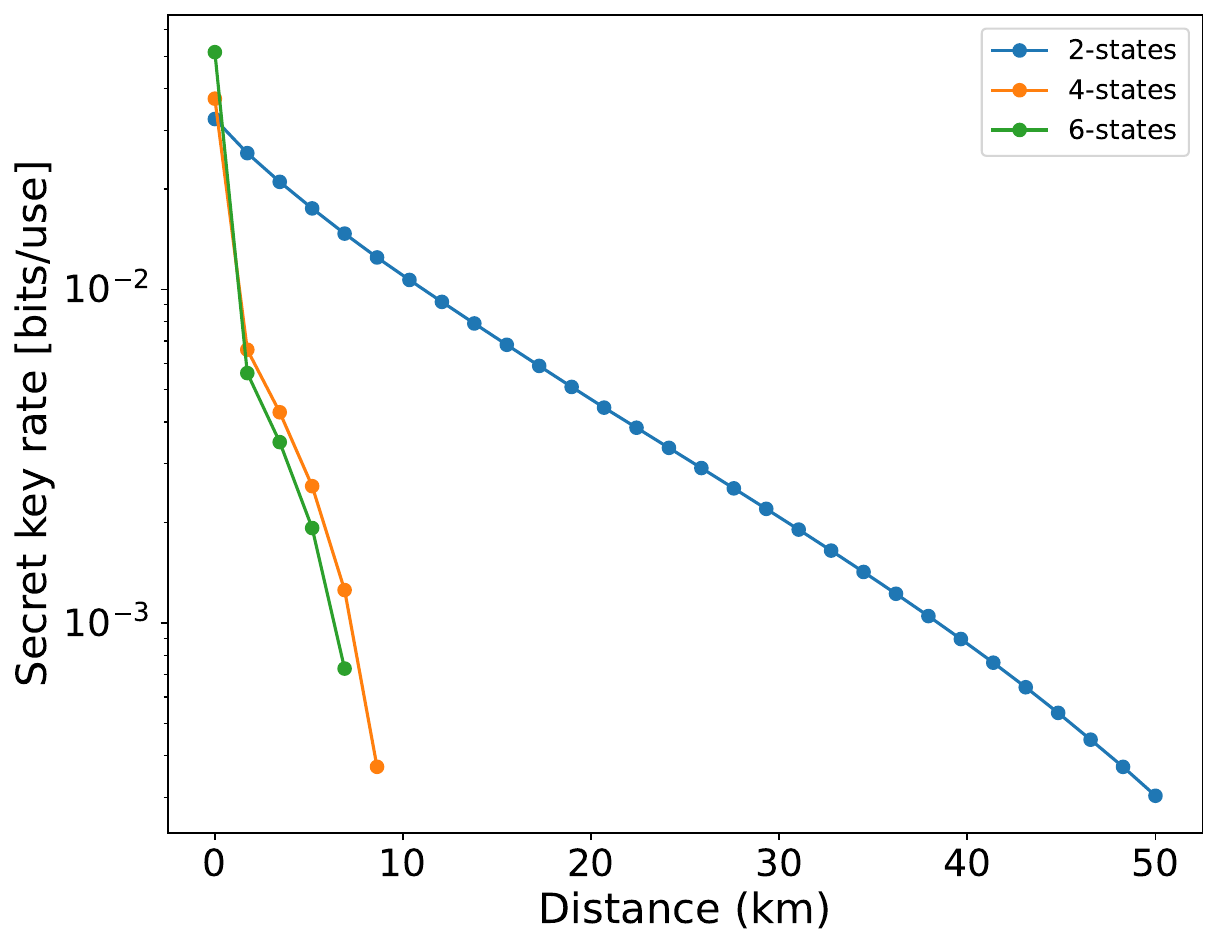}
    }

    \caption{\justifying Secret key rate for optimized variances as a function of distance for the 2-, 4-, and 6-state protocols for (a) \(\xi = 0\), (b) \(\xi = 0.005\), and $\beta = 1$. Increasing the modulation variance leads to an overestimation of the Holevo information and therefore to lower SKRs, while at low variances the SKR saturates for constellations with a small number of states.
    }

    \label{optimizedSKR}
\end{figure}

\section{Conclusion and future developments}
\label{conclusion}

We have developed a security analysis for 1D DM CV-QKD protocols under the Gaussian extremality assumption by extending the method in Ref.\cite{Ghorai2019}. By establishing the appropriate symmetry arguments for coherent states distributed along the real axis in phase space, we derived security bounds against collective attacks in the asymptotic regime via semidefinite programming. Our analysis reveals a critical limitation: the Gaussian extremality assumption systematically overestimates Eve's information in this setting, yielding bounds so conservative that secure key extraction becomes increasingly limited for constellations larger than 4 states — even under ideal conditions (see Fig.~\ref{fig:holevoinfo}) — with the SKR diminishing as the constellation size grows. This overestimation worsens significantly in the presence of excess noise. It restricts the viable range of modulation amplitudes to impractically small values near the vacuum limit, as shown in Fig.~\ref{fig:skralpha2}. In sum, such an effect is already observable for 2/4/6/8 states. Adding more states would overestimate the Holevo information even more, which we verified numerically for up to 12 states, yielding no positive key rate.

Notably, this behavior contrasts sharply with the 2D case, where the Gaussian extremality method improves as the constellation size increases~\cite{Denys2021}. In the 2D setting, larger constellations have more symmetries in phase space; that is, the constellation state becomes more isotropic and, consequently, the symmetrization process induces smaller perturbations. Consequently, the state of the constellation approaches a Gaussian state, and the Gaussian extremality assumption yields tighter bounds on Eve's information. Conversely, for the 1D case, increasing the number of modulated states in only one quadrature does not add more symmetries to the constellation state. In fact, increasing the constellation size also increases the modulation variance, and therefore the overall state becomes less similar to a Gaussian state. This implies that computing the Holevo information from the symplectic eigenvalues of the covariance matrix yields excessively loose bounds and, hence, is not appropriate for 1D protocols.

A potential alternative approach worth exploring would be to leverage the success of Gaussian extremality in 2D protocols by starting with a nearly isotropic constellation in the entanglement-based picture and then projecting it onto the 1D subspace through an appropriate choice of Alice's measurement basis. If such a construction were possible, it would allow the security analysis to inherit the favorable properties of 2D discrete modulation — where larger constellations naturally approach isotropy. This idea is somewhat related to the spirit of Ref.\cite{zhao2020unidimensional}, although the choice of measurement basis requires careful verification. This work therefore offers a potential direction for security proofs that preserve the advantages of the method proposed by Denys \textit{et al.}~\cite{Denys2021}, which provides a considerable advantage over the approach in Ref.\cite{lin2019asymptotic} by enabling an analytical solution to the SDP and maintaining high numerical efficiency as the constellation size increases.  Whether a correct measurement basis exists that achieves this projection while preserving the necessary symmetry properties remains an open question. Finally, we note that we have considered only uniformly distributed coherent states along the real axis. Allowing variable spacing or optimized amplitude distributions may improve performance even under the Gaussian extremality assumption and warrants further investigation.

\ack{We thank Marcelo Terra Cunha for insightful discussions. This work was partially funded by the project ``Security analysis of P\&M CV-QKD protocols'' supported by QuIIN - Quantum Industrial Innovation, EMBRAPII CIMATEC Competence Center in Quantum Technologies, with financial resources from the PPI IoT/Manufatura 4.0 of the MCTI grant number 053/2023, signed with EMBRAPII. MAD thanks the European Union for financing (HORIZON-MSCA-2023 Postdoctoral Fellowship, 101153602 - COCoVaQ).}



\bibliographystyle{iopart-num}
\bibliography{cvqkd}

\appendix

\section{Physicality region}
\label{physregion}
The uncertainty principle equation
\begin{align}
    \gamma_{AB} + i \Omega &\geq 0,
\end{align}
implies that
\begin{equation}
    \det(\gamma_{AB})\geq \Delta-1.
\end{equation}

By substituting Eqs.\eqref{eq:delta} and \eqref{eq:det} we can obtain the physicality region for the parameter $C_p$ as a function of the covariance matrix parameters in Eq.\eqref{CM}:
\begin{equation}
\begin{split}
    VWW_p-VWC_p^2+C_q^2C_p^2-C_q^2W_p &\geq V+WW_p+2C_qC_p-1 \\
    (VW-C_q^2)W_p-(VW-C_q^2)C_p^2 &\geq V+WW_p+2C_qC_p-1\\
    (VW-C_q^2)W_p-V-WW_p+1 &\geq (VW-C_q^2)C_p^2+2C_qC_p+\frac{C_q^2}{VW-C_q^2}-\frac{C_q^2}{VW-C_q^2}\\
    (VW-C_q^2)W_p-V-WW_p+1 &\geq \left(\sqrt{VW-C_q^2}C_p+\frac{C_q}{\sqrt{VW-C_q^2}}\right)^2-\frac{C_q^2}{VW-C_q^2}.
\end{split}
\end{equation}
Therefore, 
\begin{equation}\label{eq:physlim}
\begin{split}
    (VW-C_q^2)\left(C_p+\frac{C_q}{VW-C_q^2}\right)^2&\leq \frac{C_q^2}{VW-C_q^2}+(VW-C_q^2)W_p-V-WW_p+1\\
    \left(C_p+\frac{C_q}{VW-C_q^2}\right)^2&\leq \frac{C_q^2}{(VW-C_q^2)^2}+\frac{(VW-C_q^2)W_p-V-WW_p+1}{(VW-C_q^2)}.  
\end{split}
\end{equation}
Finally,
\begin{equation}\label{eq:physic}
    \left(C_p+C_0\right)^2\leq \left(1-\frac{W}{V}W_0\right)\left(W_p-W_0\right),
\end{equation}
with $C_0=\frac{C_q}{VW-C_q^2}$, $W_0=\frac{V}{VW-C_q^2}$.

\section{Analytical structure and numerical verification of the monotonicity of the maximum Holevo information}
\label{app:monotonicity}

The symplectic eigenvalues $\nu_{1,2}$ and $\nu_3$ depend on the correlation parameters $C_q$ and $C_p$:

\begin{equation}
  \nu_{1,2}^2 = \tfrac{1}{2}\!\left(
      \Delta \pm \sqrt{\Delta^2 - 4\det(\gamma_{AB})}
    \right),
\end{equation}
and
\begin{equation}
    \nu_3^2 = V - \frac{C_q^2}{W},
\end{equation}
with $\Delta$ and $\det(\gamma_{AB})$ defined in Eqs.~\eqref{eq:delta} and \eqref{eq:det}, respectively.

The parameters $V$, $W$, and $W_p$ are fixed from the local classical data; $C_q$ and $C_p$ are the unknown quadrature correlations.

The physical admissibility of the Gaussian state with covariance matrix $\gamma_{AB}$ imposes the constraint defined in Eq.~\eqref{eq:physic}, which defines a parabolic feasible set $\mathcal{F}(C_q)$ for $C_p$ at each fixed $C_q$.

\subsection{Definition of the reduced function}
 
For each fixed $C_q > 0$, we define the worst-case Holevo information
(maximized over physically admissible values of $C_p$):
\begin{equation}\label{eq:chimax}
  f(C_q, C_p^{\mathrm{max}})
    = \max_{C_p\in\mathcal{F}(C_q)} f(C_q, C_p).
\end{equation}
The security argument requires that $f(C_q, C_p^{\mathrm{max}})$ be
decreasing for $C_q > 0$, i.e., that increasing the positive-modulation
quadrature correlation does not benefit Eve.

\subsection{Analytical structure: why a closed-form proof is not available}

\paragraph{Interior optimum — envelope theorem.}
When the maximizer $C_p^*(C_q)$ lies in the interior of
$\mathcal{F}(C_q)$, the condition 
$\partial f/\partial C_p = 0$ is satisfied. The envelope theorem
then yields
\begin{equation}\label{eq:envelope}
  \frac{d f(C_q, C_p^{\mathrm{max}})}{dC_q}
    = \left.\frac{\partial f}{\partial C_q}\right|_{C_p=C_p^*(C_q)},
\end{equation}
so that the sign of the derivative reduces to the sign of the partial
derivative $\partial f/\partial C_q$ evaluated at the unknown point $C_p^*(C_q)$.

The explicit computation of this partial derivative requires the derivatives of $\Delta$ and $\det(\gamma_{AB})$:
\begin{align}
  \frac{\partial\Delta}{\partial C_q} &= 2C_p, \\
 \frac{\partial\det(\gamma_{AB})}{\partial C_q}
    &= 2C_q\!\left(C_p^2 - W_p\right),
\end{align}
Since the sign of $C_p$ is a priori unknown, the contribution of
$\partial\Delta/\partial C_q = 2C_p$ to $\partial\nu_{1,2}^2/\partial C_q$
may be positive or negative depending on the regime.

The term arising from $\partial\det(\gamma_{AB})/\partial C_q = 2C_q(C_p^2 - W_p)$,
on the other hand, is always negative for $C_q > 0$, since $W_p > C_p^2$
in the interior of $\mathcal{F}(C_q)$. This implies that $-4\det(\gamma_{AB})$
increases with $C_q$, enlarging the radical $\sqrt{\Delta^2 - 4\det(\gamma_{AB})}$.
As a consequence, $\nu_1^2$ increases and $\nu_2^2$ decreases, pushing the
two symplectic eigenvalues apart.

Simultaneously, $\partial\nu_3^2/\partial C_q = -2C_q/W < 0$ for $C_q > 0$,
so that increasing $C_q$ always reduces $\nu_3$. However, since $\nu_3$ enters $f$
with a minus sign through $-g\!\left((\nu_3-1)/2\right)$, this reduction contributes to an
increase in $f$. A uniform determination of the sign of
$\partial f/\partial C_q$ is therefore not analytically accessible.\\

\paragraph{Boundary optimum.}
A potentially simpler situation arises when the maximum of the Holevo information occurs on the physical boundary of the feasible region, i.e., when the physicality constraint~\eqref{eq:physic} is saturated. In this case, the covariance matrix saturates the uncertainty relation, implying
\begin{equation}
    \det(\gamma_{AB})=\Delta - 1.
\end{equation}

As a consequence, the global symplectic spectrum simplifies considerably:
\begin{equation}
    \nu_1 = \sqrt{\det(\gamma_{AB})},
    \qquad
    \nu_2 = 1.
\end{equation}
Thus, the Holevo information reduces to
\begin{equation}
   f(C_q,C_p)= g\!\left(\frac{\nu_1-1}{2}\right)
    -
    g\!\left(\frac{\nu_3-1}{2}\right).
\end{equation}

This boundary reduction eliminates one global symplectic contribution and transforms the problem into a competition between only two terms: a global term, depending on $\det(\gamma_{AB})$, and a conditional term, depending explicitly on $C_q$.
Although this structure is significantly simpler than the interior case, a closed-form analytical proof of monotonicity,
\begin{equation}
    \frac{d f(C_q,C_p^{\mathrm{max}})}{dC_q}<0,
\end{equation}
still requires strong and non-trivial assumptions, which can only be guaranteed through numerical analysis:
\begin{enumerate}
    \item It is necessary to determine which branch, $C_p^+(C_q)$ or $C_p^-(C_q)$,
          actually maximizes the Holevo information. We observe that this typically occurs at the upper branch;
    \item It is necessary to establish the monotonic behavior of the maximizing branch
          $\frac{dC_p^*}{dC_q}$, whose sign depends on the specific geometry of the
          boundary and is not universal. We observe that it is generally decreasing;
    \item It is necessary to quantitatively compare the competing contributions of the
          derivatives of $\nu_1$ and $\nu_3$; while $\frac{d\nu_3}{dC_q}<0$ is obtained
          explicitly, the derivative of $\nu_1$ depends implicitly on the entire boundary
          geometry, which remains algebraically complex.
\end{enumerate}

Therefore, although the boundary saturation regime substantially simplifies the symplectic structure, the analytical proof of monotonicity of the maximum Holevo information still depends on additional strong assumptions that generally prevent a global proof without further restrictions on the physical parameters.

In summary, in both regimes the analytical structure considerably hinders a global closed-form proof. This motivates the numerical study presented below.
 
\subsection{Numerical verification strategy}\label{app:Numerical_verificatio}
 
We verify the monotonicity of $f(C_q,C_p^{\mathrm{max}})$ numerically
over the physically relevant parameter space. The strategy consists of two steps:

\begin{enumerate}
    \item \textbf{Discretization of parameter space:}
A grid of experimentally motivated values $(V, W, W_p)$ is fixed. For each triple
and each value of $C_q$ on a fine uniform grid over $[0, C_q^{\mathrm{max}}]$, where
$C_q^{\mathrm{max}}$ is determined by the positivity of the right-hand side
of~\eqref{eq:physic}, the maximization $f(C_q,C_p^{\mathrm{max}})$ is solved
within the analytically derived feasibility bounds from~\eqref{eq:physic}.

\item \textbf{Monotonicity check:}
For each parameter triple, it is verified that the sequence
$f(C_q^{(k)},C_p^{\mathrm{max}})$ is decreasing along the grid.
\end{enumerate}
 
The numerical study confirms that $f(C_q,C_p^{\mathrm{max}})$ is monotonically
decreasing for $C_q > 0$, consistently with the expectation that a stronger signal correlation reduces Eve's informational advantage. No violation of monotonicity was observed in any of the
tested configurations.

\section{Details on the numerical methodology}
\label{app:methodology}

The numerical analysis of the 1D GM CV-QKD protocol was implemented directly from the results of Ref.\cite{Usenko2015}, while the pure-loss case was computed from the numerical integration of the definitions of the mutual and Holevo information.

The Gaussian extremality argument employed in our work consists simply in the application of the Gaussian extremality theorem, which states that among all quantum states sharing the same first and second moments, i.e., the same covariance matrix, Gaussian states minimize the secret key rate against collective attacks \cite{wolf2006}. In the context of CV-QKD security, this property allows us to upper-bound Eve's information (Holevo information) by computing it as if the transmitted state were Gaussian. This approach transforms the complicated problem of determining Eve's optimal attack into an optimization problem over the unknown parameters of the covariance matrix \cite{Ghorai2019, Denys2021}. In our case, there are five parameters in the covariance matrix: $V, W, W_p, C_p$, and $C_q$. Once the Gaussian state extremality theorem is invoked, the Holevo information is computed from its known closed-form formula expressed in terms of the symplectic eigenvalues of that covariance matrix.

However, two out of the five covariance matrix parameters remain unknown — namely, the correlations in the modulated and unmodulated quadratures, $C_q$ and $C_p$, respectively. Without determining these parameters, the Holevo information cannot be computed. The key point is that, although $C_q$ and $C_p$ are not directly accessible, their physically admissible values are constrained by Alice and Bob's measurements statistics and the Heisenberg uncertainty principle \cite{Serafini_2005,Usenko2015}, which defines the so-called physicality zone. Within this restricted feasible set, the worst-case values of $C_q$ and $C_p$ — those that maximize Eve's information while remaining consistent with Alice and Bob's observed measurement statistics — are identified by casting the problem as a semidefinite program (SDP). In this sense, Gaussian extremality and the SDP are sequential and complementary steps: the former justifies replacing the true, unknown quantum state with a Gaussian proxy sharing the same covariance matrix, while the latter determines which proxy represents the worst case for security, yielding a provable upper bound on Eve's accessible information. The necessity of employing SDP techniques in the security analysis of discretely-modulated CV-QKD protocols is well-known in the field~\cite{zhang2024continuous,usenko2026continuous,anka2026introductory,Ghorai2019,Denys2021,lin2019asymptotic}. 

Below, we provide more details on the physical and computational parameters used in our simulations of Figs.~\ref{fig:physicality_region}-\ref{optimizedSKR}.  

\subsection{Physicality region (Fig.~\ref{fig:physicality_region}) and maximum Holevo information
            vs.\ correlation (Fig.~\ref{fig:Holevo maximum})}

Figures~\ref{fig:physicality_region} and~\ref{fig:Holevo maximum} require no SDP solve. For the fixed values of $(\alpha_0, d, \xi)$, we computed the modulation variance $V=1+4\alpha_0^2\sum_k p_k r_k^2$ and noise variances given by the channel model $W=T(V-1)+1+T\xi$ and $W_p=1+T\xi$ with $T=10^{-0.02d}$. For each of 1200 uniformly spaced values of $C_q$ within the physical cap $|C_q|\leq\sqrt{VW-W}$, the bounds $C_p^{\pm}$ are evaluated in closed form from Eq.\eqref{eq:cp-}. For Fig.~\ref{fig:Holevo maximum}, $\chi(Y;E)$ is additionally maximized over 300 uniformly
spaced values of $C_p\in[C_p^-,C_p^+]$ using Eqs.\eqref{eigenvalues}--\eqref{eq:holevo}. The parameters are specified in the tables below.

\subsection{SDP-based bounds (Figs.~\ref{mutualholevoinfopure}--\ref{optimizedSKR})}

The SDP minimizes $|C_q(\rho)|$ as stated in Eq.\eqref{sdp},
implemented in CVXPY using MOSEK (version~11.0.29).
Coherent states $|\pm\alpha_k\rangle$ are represented in a truncated
Fock basis of dimension $\mathrm{dim}=10$. We tested increasing the dimension, but no improvement was observed.
The ensemble average state $\tau$ and its pseudo-inverse square root
$\tau^{-1/2}$ are computed by singular value decomposition (SVD) with
a singular-value cutoff of $10^{-8}$ for numerical stability.
A second SVD with cutoff $10^{-12}$ identifies the effective support
of the projector $P$ (see Sec.~\ref{sec:purification}) and reduces the
Alice space to dimension $r_{\mathrm{red}}=2N$.
The bipartite SDP variable acts on a complex space of dimension
$2N\times\mathrm{dim}$.
Complex Hermitian operators are mapped to real symmetric matrices by
the realification
$A=X+iY\mapsto\bigl(\begin{smallmatrix}X&Y\\-Y&X\end{smallmatrix}\bigr)$,
and each operator is Frobenius-normalized before passing to MOSEK.
The partial-trace constraint
$\mathrm{Tr}_B(\rho_{AB})=\tau_{\mathrm{red}}$ is enforced as hard
linear equalities with box tolerance $10^{-8}$.
After each solve, the feasibility interval $[C_p^-,C_p^+]$ is
computed from Eq.\eqref{eq:parabola}, and the worst-case $C_p^*$
maximizing $\chi(Y;E)$ is selected by a 300-point uniform grid scan.

For Figs.~\ref{mutualholevoinfo} and~\ref{optimizedSKR}, the
modulation amplitude $\alpha_0$ is optimized at each distance by a
brute-force grid search over 25 uniformly spaced values of
$\alpha_0^2\in[10^{-4},10^{-1}]$, with an independent SDP solve at each
grid point. The optimal amplitude is the grid point maximizing
$K=\beta I(X;Y)-\chi(Y;E)$.

\subsection{Uniform vs discrete Gaussian probability distributions}

We attempted to improve the performance of the 1D DM CV-QKD protocol through probabilistic shaping. Figure~\ref{gaussiandiscrete} shows the SKR for the 4-state protocol with $\xi = 0$ and $\beta = 1$ under uniform and discrete Gaussian probability distributions. In the latter case, the probability assigned to ring $k$ is $p_k \propto \exp\!\left(- r_k^2/2s^2\right), \quad k = 0, 1, \ldots, N-1$, where $s > 0$ is the standard deviation of the distribution, and equal probability $p_k/2$ is assigned to each of the two states $|\pm\alpha_k\rangle$ on that ring. In the limit $s \to 0$, all weight concentrates on the innermost
ring, and the protocol reduces to BPSK. In the limit $s \to \infty$, all rings become equally weighted and the distribution recovers the uniform case. As shown in Fig.~\ref{gaussiandiscrete}, probabilistic shaping does not improve the SKR: all discrete Gaussian curves lie between the uniform distribution and the BPSK case, confirming that the uniform distribution performance is better under the Gaussian extremality assumption for this class of protocols. The threshold point where the uniform distribution becomes worse than the discrete Gaussian case is due to the limitation of our approach. One can observe that the crossing point of the uniform and discrete Gaussian curves, for $s=0.1$ (essentially a BPSK system), is due to the overestimation of the Holevo information, which increases systematically as the number of states increases (see Fig.~\ref{fig:holevoinfo}).

\begin{figure}[h!]
    \centering
    \includegraphics[width=0.6\textwidth]{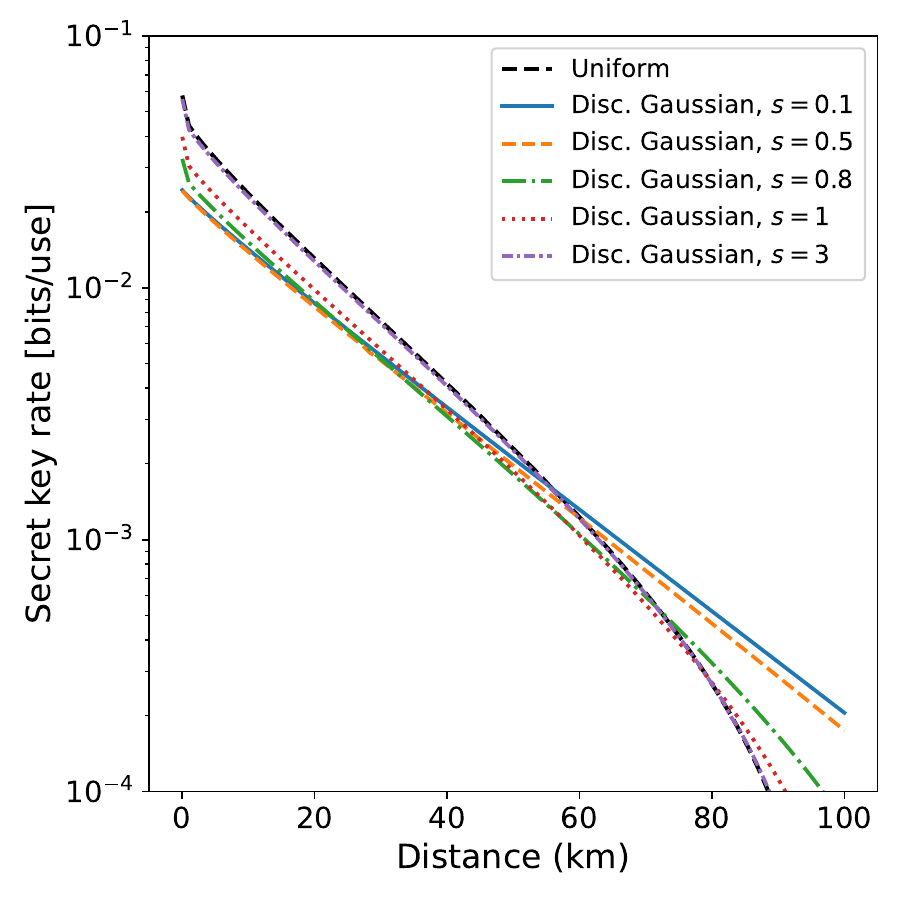}
    \caption{\justifying Comparison between the 1D DM CV-QKD protocol with uniform and discrete Gaussian probability distributions for a 4-state constellation with \(\xi = 0\) and $\beta = 1$.}
    \label{gaussiandiscrete}
\end{figure}


\begin{table*}[ht]
\caption{%
\justifying Protocol parameters shared by all figures.
$N$ is the number of ring pairs; the total number of coherent states
is $2N$.
SNU: shot-noise unit (vacuum variance $=1$).}
\label{tab:shared}
\small
\renewcommand{\arraystretch}{1.4}
\begin{tabular*}{\textwidth}{@{\extracolsep{\fill}}
  L{3.5cm} L{5.2cm} L{7.8cm}}
\toprule
\textbf{Parameter} & \textbf{Value} & \textbf{Description} \\
\midrule
Fiber loss coefficient
  & $\gamma = 0.2$~dB/km
  & Transmittance $T=10^{-\gamma d/10}$ \\
Quadrature convention
  & $\hbar=2$
  & $[\hat{q},\hat{p}]=2i$;\; $\hat{q}=\hat{a}+\hat{a}^\dagger$,\;
    $\hat{p}=i(\hat{a}^\dagger-\hat{a})$ \\
Amplitude ratios
  & $r_k=k + 1$,\; $k=0,\ldots,N$
  & Integer ring radii; amplitudes $\alpha_k=\alpha_0 r_k$ \\
Probability distribution
  & Uniform: $p_k=1/N$
  & All $N$ rings equally weighted \\
Modulation variance
  & $V=1+4\alpha_0^2\sum_{k=1}^{N}p_k r_k^2$
  & See Sec.~\ref{correlationsection} \\
Bob $\hat{q}$ variance
  & $W=T(V-1)+1+T\xi$
  & Variance in the modulated quadrature \\
Bob $\hat{p}$ variance
  & $W_p=1+T\xi$
  & Variance in the unmodulated quadrature \\
Correlation target
  & $c_{2N}(T,\xi)=2\sqrt{T}\alpha_0\sum_{k=0}^{N-1}p_kr_k$
  & See Sec.~\ref{correlationsection} \\
Reconciliation efficiency
  & $\beta = 1$
  & Perfect reconciliation \\
\bottomrule
\end{tabular*}
\end{table*}


\begin{table*}[ht]
\caption{%
\justifying Figure-specific parameters.
\textbf{Analytic}: closed-form evaluation of
Eqs.~\eqref{eq:parabola}--\eqref{eq:holevo}; no solver required.
\textbf{NI}: direct numerical integration of the pure-loss channel.
\textbf{SDP}: solution of the semidefinite program, Eq.~\eqref{sdp}.
The modulation variance $V$ is computed from $\alpha_0$ via
Table~\ref{tab:shared}.
The $C_q$ grid in
Figs.~\ref{fig:physicality_region}--\ref{fig:Holevo maximum}
spans $[-\sqrt{VW-W},+\sqrt{VW-W}]$ with 1200 points. In Figs.~\ref{fig:skrcombined} and \ref{fig:skrcombined2}, ``$\pm p$~est.'' denotes two separate runs: with and without
$\hat{p}$-quadrature correlation estimation.
For Figs.~\ref{mutualholevoinfo} and \ref{optimizedSKR}
(optimized $\alpha_0$), $V$ is re-derived at each grid point.}
\label{tab:figparams}
\small
\renewcommand{\arraystretch}{1.4}
\begin{tabular*}{\textwidth}{@{\extracolsep{\fill}}
  C{1.0cm} C{1.8cm} C{2.0cm} C{2.0cm} C{1.3cm} C{1.1cm} C{0.7cm} L{2.4cm}}
\toprule
\textbf{Figure}
  & \textbf{Method}
  & \textbf{States ($2N$)}
  & \textbf{$V$}
  & \textbf{$\alpha_0$}
  & \textbf{$\xi$}
  & \textbf{$\beta$}
  & \textbf{Distance / range} \\
\midrule
Fig.~\ref{fig:physicality_region}
  & Analytic
  & 2, 4 & 1.04, 1.10 & 0.1 & 0 & ---
  & $d=10$~km; $C_q$ grid \\
\midrule
Fig.~\ref{fig:Holevo maximum}
  & Analytic
  & 2, 4, 6 & 1.04, 1.10, 1.18 & 0.1 & 0 & ---
  & $d=10$~km; $C_q$ grid \\
\midrule
Fig.~\ref{mutualholevoinfopure}
  & NI / SDP
  & 2, 4, 6 & 1.04, 1.10, 1.18 & 0.1 & 0 & ---
  & $0.01$--$100$~km, 100 pts \\
\midrule
Fig.~\ref{fig:skrcombined}
  & NI / SDP
  & 2, 4 & 1.04, 1.10 & 0.1 & 0, 0.005 & 1
  & $0.01$--$100$~km, 100 pts; $\pm p$~est.\\
\midrule
Fig.~\ref{fig:skrcombined2}
  & SDP
  & 4 & 1.10 & 0.1 & 0 & 1
  & $0.01$--$14$~km, 100 pts; $\pm p$~est. \\
\midrule
\multirow{2}{*}{Fig.~\ref{mutualholevoinfo}}
  & \multirow{2}{*}{SDP}
  & 2, 4, 6, 8 & opt. & opt. & 0 & 1
  & fixed $d=10$~km; 25-pt $\alpha_0^2$ grid \\
  & & 2, 4, 6, 8 & opt. & opt. & 0.005 & 1
  & fixed $d=2$~km; 25-pt $\alpha_0^2$ grid \\
\midrule
Fig.~\ref{optimizedSKR}
  & SDP
  & 2, 4, 6 & opt. & opt. & 0, 0.005 & 1
  & $0.01$--$50$~km, 100 pts; opt.\ $\alpha_0$ \\
\bottomrule
\end{tabular*}
\end{table*}

\end{document}